%% file: NF_PLS_draft.tex
\documentclass[lettersize,journal]{IEEEtran}
\usepackage{amsmath,amssymb}
\usepackage{subfigure}
\usepackage{graphicx, graphics, color,psfrag}
\usepackage{cite,balance}
\usepackage{caption}
\allowdisplaybreaks
\usepackage{algorithm}
\usepackage{algorithmic}
\usepackage{accents}
\usepackage{amsthm}
\usepackage{bm}
\usepackage{url}
\usepackage[english]{babel}
\usepackage{multirow}
\usepackage{enumerate}
\usepackage{cases}
\usepackage{stfloats}
\usepackage{dsfont}
\usepackage{color,soul}
\usepackage{amsfonts}
\usepackage{cite,graphicx,amsmath,amssymb}
\usepackage{subfigure}
\usepackage{fancyhdr}
\usepackage{hhline}
\usepackage{graphicx,graphics}
\usepackage{array,color}
\usepackage{mathtools}
\usepackage{amsmath}

\include{header}

\begin{document}
\captionsetup[figure]{name={Fig.},labelsep=period,singlelinecheck=off}  
\title{\huge 
Performance Analysis and Low-Complexity Beamforming Design for Near-Field Physical Layer Security\vspace{-8pt}}
\author{Yunpu~Zhang,~\IEEEmembership{Graduate Student Member,~IEEE}, Yuan~Fang,~\IEEEmembership{Member,~IEEE},
Changsheng~You,~\IEEEmembership{Member,~IEEE}, Ying-Jun Angela Zhang,~\IEEEmembership{Fellow,~IEEE}, and Hing Cheung So,~\IEEEmembership{Fellow,~IEEE}\vspace{-0.5em}
\thanks{
Y. Zhang, Y. Fang and H. C. So are with the Department of Electrical Engineering, City University of Hong Kong, Hong Kong (e-mail:
yunpu.zhang@my.cityu.edu.hk, yuanfang@cityu.edu.hk, hcso@ee.cityu.edu.hk). C. You is with the Department of Electronic and Electrical Engineering, Southern University of Science and Technology (SUSTech), Shenzhen
518055, China (e-mail: youcs@sustech.edu.cn). Y.-J. A. Zhang is with the Department of Information Engineering, The Chinese University of Hong Kong, Hong Kong (e-mail: yjzhang@ie.cuhk.edu.hk).}  
\vspace{-18pt}} 
% \emph{(Corresponding
% author: Xianghao Yu.)
% C. You is with the Department of Electrical and Electronic Engineering, Southern University of Science and Technology, Shenzhen
% 518055, China (e-mail: youcs@sustech.edu.cn)

\maketitle
	\begin{abstract}
	Extremely large-scale arrays (XL-arrays) have emerged as a key enabler in achieving the unprecedented performance requirements of future wireless networks, leading to a significant increase in the range of the near-field region. This transition necessitates the spherical wavefront model for characterizing the wireless propagation rather than the far-field planar counterpart, thereby introducing extra degrees-of-freedom (DoFs) to wireless system design. In this paper, we explore the beam focusing-based physical layer security (PLS) in the near field, where multiple legitimate users and one eavesdropper are situated in the near-field region of the XL-array base station (BS). First, we consider a special case with one legitimate user and one eavesdropper to shed useful insights into near-field PLS. In particular, it is shown that 
 %AN can benefit PLS improvement in the near field from two aspects: 
 1) Artificial noise (AN) is crucial to near-field \emph{security provisioning}, transforming an insecure system to a secure one; 2) AN can yield numerous \emph{security gains}, which considerably enhances PLS in the near field as compared to the case without AN taken into account.  Next, for the general case with multiple legitimate users, we propose an efficient low-complexity approach to design the beamforming with AN to guarantee near-field secure transmission. Specifically, the low-complexity approach is conceived starting by introducing the concept of \emph{interference domain} to capture the inter-user interference level, followed by a \emph{three-step identification framework} for designing the beamforming. Finally, numerical results reveal that 1) the PLS enhancement in the near field is pronounced thanks to the additional spatial DoFs; 2) the proposed approach can achieve close performance to that of the computationally-extensive conventional method yet with a significantly lower computational complexity.
\end{abstract}
\vspace{-10pt}
	\section{Introduction}
To fulfill the growing demands for applications like extended reality, metaverse, and holographic telepresence, the future sixth-generation (6G) wireless networks are anticipated to deliver significantly enhanced performance requirements than those of fifth-generation (5G) \cite{8869705,you2024next}. Among others, by drastically boosting the number of antennas deployed at the base station (BS),  the extremely
large-scale array (XL-array) has been envisioned as a key enabler to significantly improve the spectral efficiency and spatial resolution of future wireless communication systems \cite{liu2023near,9903389,you2023near}.
In particular, the huge aperture of XL-array introduces a fundamental change in the electromagnetic channel modeling, leading to a paradigm shift from traditional far-field communications towards new \emph{near-field communications} \cite{9903389}. 
More specifically, the electromagnetic waves in the near-field region should be precisely modeled by the spherical wavefront, which differs from the planar one typically assumed in the far-field counterpart. 

This unique characteristic facilitates
 the beam energy to be focused at a specific location/region, termed \emph{beam focusing} \cite{9738442}, bringing new design opportunities in various wireless systems \cite{9617121,wang2023beamfocusing,LDMA,zhang2023swipt}. 
 For example, for the typical multi-user communication systems, the authors of \cite{9617121} showed that using maximum ratio transmission (MRT)-based transmit beamforming can lead to near-optimal performance when the number of antennas is sufficiently large, thanks to the beam focusing gain. Later on, the beam focusing property was further investigated in near-field wideband multi-user communication systems for alleviating the spatial wideband effect \cite{wang2023beamfocusing}. In \cite{LDMA}, a novel location division multiple access capitalizing on the beam focusing property was proposed to serve multiple near-field users simultaneously to enhance the spectrum efficiency. 
 Meanwhile, 
 % the near-field spherical wavefront introduces two degrees of freedom (DoFs) concerning both angle and distance, 
 the beam focusing characteristic 
stimulates plentiful
 novel applications, especially in the realm of sensing/localization-related scenarios. For instance, power in near-field simultaneous wireless information
and power transfer (SWIPT) systems can be intentionally concentrated on the power-hungry receivers with very limited energy leakage rather than being dispersed along a specific direction as in the far-field SWIPT \cite{zhang2023swipt}. 
 % Furthermore, by properly exploiting the electromagnetic wavefront curvature, both angle and distance information can be inferred for near-field tracking using only a single XL-array, which is unachievable with its far-field counterpart \cite{9508850}. 
 Besides, the authors of \cite{10520715} developed a near-field integrated sensing and communications framework for multi-target detection. Despite the prominent merits of the beam focusing characteristic in various communication scenarios, security has been a long-standing concern in wireless communication systems, which deserves more in-depth research in the new near-field communications.

Physical layer security (PLS) is an effective complement to cryptography for safeguarding system security  \cite{shiu2011physical}, making use of the physical attributes of
wireless channels, such as interference, fading, noise, and disparity, while avoiding the complexity of generating and managing secret keys. 
% Accordingly, a question naturally arises: {Is the unique beam focusing characteristic beneficial to PLS in near-field communications?} More importantly, it is of great importance to characterize the new PLS features encapsulated in near-field spherical-wave propagation.  
%However, the existing works on PLS have been widely adopted the planar-wave channel model \cite{7812773,8972400,8723525}, which may be inapplicable to achieve PLS in near-field communications. In addition, 
There have been some initial attempts made to investigate PLS in the near-field communications \cite{zhang2023near,10480457,ferreira2023physical,9860861,yu2022physical,zhang2023physical}. The work in \cite{zhang2023near} and \cite{10480457} explored analog beam focusing for near-field PLS enhancement, considering the hardware costs associated with XL-arrays.
Specifically,
in \cite{zhang2023near}, the PLS improvement in near-field wideband communications was investigated, wherein the interplay between near-field
propagation and wideband beam split were addressed by utilizing the true-time delayer (TTD)-based analog beam focusing. The authors of \cite{10480457} proposed an analog directional modulation precoding algorithm to offer secure transmissions at both angular and distance scales for a single-user multiple-input single-output (MISO) system. Further studies have demonstrated various benefits of beam focusing in near-field secure transmission. In particular, it was revealed in \cite{ferreira2023physical} that the beam focusing effect can be utilized to significantly improve the jamming rejection and the secrecy performance in near-field communications. The authors of \cite{9860861} showed through numerical results that utilizing the new spatial degree-of-freedom (DoF) under spherical-wave propagation can enhance the secrecy rate and reduce the secure-blind areas. In \cite{yu2022physical}, the authors studied PLS in the spherical-wave channel model, in which a secure transmission scheme was suggested to combat the cooperative eavesdropping for achieving distance-domain security. 
% Besides, considering the benefits of spherical-wavefront propagation to PLS, the authors of \cite{9967979} proposed a leakage subspace precoding scheme for joint secure precoding and user scheduling, achieving a significant secrecy spectral efficiency enhancement compared to the conventional zero-forcing methods. 
Moreover, the authors of \cite{zhang2023physical} introduced a near-field secure transmission framework using the classical hybrid beamforming architecture, revealing that the secrecy transmission of near-field communications is mainly determined by the relative distance of the eavesdropper with respect to the legitimate user. 

% \begin{figure}[t]
% 	\centering
% 	\includegraphics[width=0.45\textwidth]{system_model1.jpg}
% 	\caption{A near-field secure communication system.} \label{Fig:systemmodel}
% \end{figure}

However, most of the existing works either numerically demonstrated the performance improvement by exploiting the beam focusing effect or directly applied conventional far-field designs to near-field secure transmission. This thus fails to fully unveil the potential of near-field PLS, highlighting the need for more in-depth research on the sophisticated secure transmission design. For example, artificial noise (AN) has been extensively utilized in conventional far-field secure communications for PLS enhancement, which is designed deliberately to impair the channel of the eavesdropper while at the same time having a limited effect on the signal-to-interference-plus-noise ratio (SINR) at the legitimate users \cite{9133130,9024490}. However, it still remains unknown whether AN is advantageous to the near-field beam focusing-based PLS, and the analytical secrecy performance characterizations of near-field PLS remain lacking. More specifically, the following aspects are worthy of investigation: 1) whether the AN can utilize the beam focusing property embedded in near-field spherical wavefront; 2) under what conditions is AN most advantageous.
% Besides, the fundamental secrecy performance characterizations of near-field PLS remain lacking. It has been shown that artificial-noise (AN)-based beamforming \cite{} is effective in improving the PLS of conventional far-field communication.   
Moreover, although the beam focusing function is appealing for its flexibility in controlling the beam energy distribution, its practical realization is typically accompanied by prohibitively high computational complexity. For example, to fully exploit the potential of beam focusing, semidefinite programming is commonly adopted for designing the beamformer, whose computational complexity is on the order of $\mathcal{O}(N^{6.5})$ \cite{10520715} with $N$ denoting the number of antennas. This is far from practical, especially in the XL-array systems. Thus, it necessitates computationally-efficient designs tailored for near-field PLS. To the best of our knowledge, the above two aspects have not yet been studied
in the literature.

 Motivated by the above, in this paper, we consider the security provisioning for a near-field communication system, where the BS equipped with an XL-array transmits confidential information to multiple near-field legitimate users in the presence of one near-field eavesdropper. Specifically, we delve into a challenging scenario in which the eavesdropper possesses the most favorable channel condition, under which joint beamforming with AN for achievable secrecy rate maximization is studied. The main contributions are summarized as follows.
 \begin{itemize}
    \item First, we investigate secrecy provisioning in a near-field multi-user communication system using the beam focusing effect, where the achievable secrecy rate is maximized by jointly optimizing the transmit beamforming with AN subject to the transmit power constraint at the BS.
	\item Second, to the best of our knowledge, we are the first to analytically investigate the effect of AN on near-field PLS. Specifically, to shed important insights into the new characteristics of near-field PLS, we consider a special case with one legitimate user and one eavesdropper involved. In particular, it is shown that incorporating AN into the beam focusing-based PLS can bring two prominent benefits. On the one hand, AN is essential to \emph{security provisioning}, being capable of transforming an insecure system into a secure one. On the other hand, we reveal an interesting fact that allocating only a small proportion of power to AN can lead to significant \emph{security gain enhancement} compared to the case without AN taken into account. 
    \item Third, for the general case with multiple legitimate users, we further develop an efficient yet low-complexity approach compared to the classical approach capitalizing on the semidefinite relaxation (SDR) and successive convex approximation (SCA).
    Specifically, the proposed approach starts with defining an \emph{interference domain} to conceptualize the interference level among users, followed by a customized \emph{three-step identification framework} for determining the beamforming optimization.
 	% \item Third, we commence with proposing a conventional approach for designing the transmit beamforming with AN, followed by an implementation-friendly low-complexity approach based on the analytical results. For  are employed to obtain a suboptimal design. For the low-complexity approach, we begin with introducing the concept of \emph{interference domain}, and then utilize a \emph{three-step} identification framework for beamforming design.
 	\item Finally, numerical results are provided to verify the performance gains attained by near-field secure transmission over the conventional far-field secure communication system, as well as the effectiveness of the proposed low-complexity approach. In particular, it is demonstrated that 1) the AN can make good use of the beam focusing effect; 2) the extra spatial resources in the distance domain are highly beneficial to near-field PLS improvement; 
    3) our approach achieves very close secrecy performance to that of the conventional approach but with much lower computational cost.
 \end{itemize}
 
% The remainder of this paper is organized as follows. Section \ref{Sec:SM} presents the system model for the near-field PLS system. In Section \ref{Sec:SC}, we consider a special case to shed useful insights into the near-field PLS. Section \ref{Se:Low-com} delves into proposing an efficient low-complexity approach for beamforming design.
% Simulation results are presented in
%  Section \ref{Sec:SR} to evaluate the performance of the proposed low-complexity scheme,
%  followed by the conclusions drawn in Section \ref{Se:Con}.
 
 \emph{Notations:} We adopt the upper-case
 calligraphic letters, e.g., $\mathcal{K}$ to denote discrete and finite sets. The superscript $(\cdot)^H$ denote the conjugate transpose. $\mathcal{CN}(\mu,\sigma^2)$ represents the complex Gaussian distribution with mean $\mu$ and variance $\sigma^2$. Moreover, $|\cdot|$ denotes the absolute value for a real number and the cardinality for a set. $\mathrm{Tr}(\mathbf{X})$ and $\mathrm{Rank}(\mathbf{X})$ denote the trace and rank of matrix $\mathbf{X}$; $\mathbf{X} \succeq 0$ indicates that $\mathbf{X}$ is a positive semidefinite (PSD) matrix. $[x]^+$ stands for $\max\{0, x\}$. $\mathcal{O}(\cdot)$ denotes the standard big-O notation. 
 \vspace{-6pt}
\section{System Model}\label{Sec:SM}
We consider a near-field multi-user downlink communication system as shown in Fig.~\ref{Fig:systemmodel1}, consisting of an XL-array BS, an eavesdropper, and multiple legitimate users, indexed by $\mathcal{K}=\{1,2,\cdots,K\}$. The BS employs a uniform linear array (ULA) composed of $N$ antenna elements, while the legitimate users and the eavesdropper are single-antenna receivers. We assume that the legitimate users and eavesdropper are located in the near-field region, where the BS-user distances are greater than the Fresnel distance $r_{\rm Fre}=0.62\sqrt{D^3/\lambda}$ and smaller than the Rayleigh distance $r_{\rm Ray}={{2D^2/}{\lambda}}$, with $D$ and $\lambda$ denoting the antenna array aperture and carrier wavelength, respectively. Moreover, perfect channel state information (CSI) of the legitimate users and the eavesdropper is assumed to be available at the BS via effective near-field channel estimation and/or beam training methods, e.g., \cite{liu2024sensing,zhang2022fast,wu2023two}.\footnote{Notice that the acquisition of the CSI of the passive eavesdropper is an important and challenging topic. While obtaining the CSI of the passive eavesdropper is generally difficult, the
results presented in this paper serve as theoretical upper bounds on the performance of the considered system. Furthermore, even for a passive eavesdropper, its CSI can potentially 
be estimated by exploiting the local oscillator power that inadvertently leaks from its receiver RF frontend \cite{6288501}.}  
	
\begin{figure}[t]
	\centering
	\includegraphics[width=0.32\textwidth]{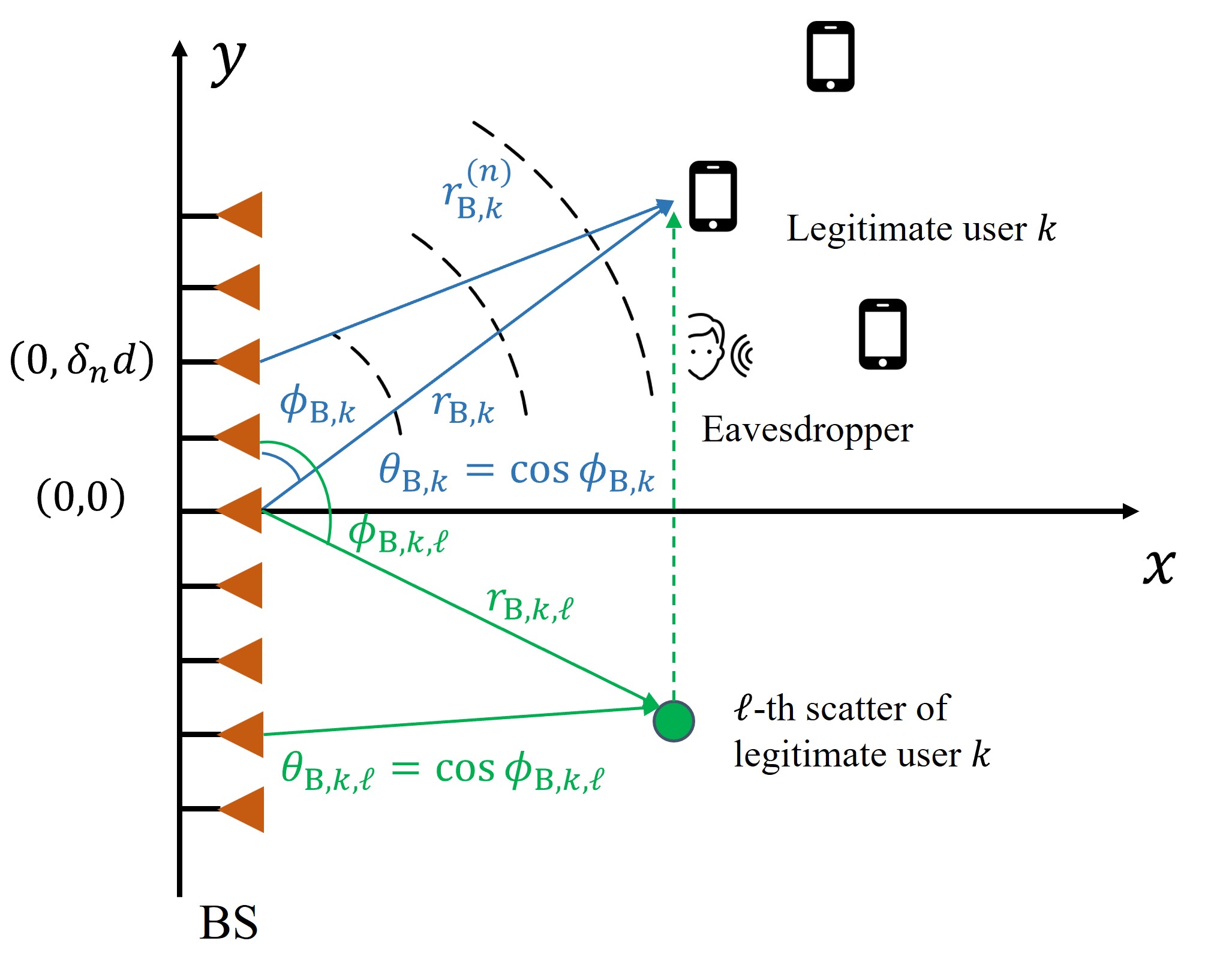}
	\caption{The considered near-field PLS system.} \label{Fig:systemmodel1}
 \vspace{-20pt}
\end{figure}
 \vspace{-12pt}
\subsection{Channel Model}
For a typical near-field user, its channel should be modeled by the more accurate spherical wavefront, which is explained as follows.
First, without loss of generality, we assume that the XL-array is placed along the $y$-axis, centered at the origin of the coordinate system, i.e., $(0,0)$. Accordingly, the coordinate of the $n$-th antenna element, $\forall n\in\{1,\cdots,N\}$, is given by $(0,\delta_nd)$, where $\delta_n=\frac{2n-N-1}{2}$, and $d=\frac{\lambda}{2}$ denotes the antenna spacing. As such, the distance between the $n$-th antenna element of the XL-array and the legitimate user $k$ can be expressed as 
\begin{equation}
	r^{(n)}_{{\rm B},k}=\sqrt{ r^{2}_{{\rm B},k}-2r_{{\rm B},k}\theta_{{\rm B},k}\delta_{n}d+\delta_{n}^2d^2},
    \vspace{-6pt}
\end{equation}
where $r_{{\rm B},k}$ and $\theta_{{\rm B},k}=\cos
\phi_{{\rm B},k}$ denote the propagation distance and spatial angle from the origin to the legitimate user $k$, respectively, with $\phi_{{\rm B},k}$ being the physical angle-of-departure (AoD). In this paper, we adopt the general near-field multipath channel model for each user, consisting of a line-of-sight (LoS) path and several non-line-of-sight (NLoS) paths induced by environment scatterers. Specifically, the LoS component of the near-field  channel between the BS and the legitimate user $k$ can be modeled as 
\begin{equation}
	\mathbf{h}^{H}_{{\rm LoS},{\rm B},k} = \sqrt{N}h_{{\rm B},k}\mathbf{b}^{H}(\theta_{{\rm B},k},r_{{\rm B},k}),
\end{equation}
where $h_{{\rm B},k}=\frac{\lambda}{4\pi r_{{\rm B},k}}e^{-j\frac{2\pi}{\lambda}r_{{\rm B},k}}$ denotes the complex-valued channel gain from the origin to the legitimate user $k$; $\mathbf{b}(\theta_{{\rm B},k},r_{{\rm B},k})$ is the channel steering vector, given by
\begin{align}\label{Eq:NF_steering}
	&\mathbf{b}^H(\theta_{{\rm B},k},r_{{\rm B},k})\nn\\
	&\!\!=\!\!\frac{1}{\sqrt{N}}\!\!\left[e^{-j 2 \pi(r^{(0)}_{{\rm B},k}-r_{{\rm B},k})/\lambda}, \cdots, e^{-j 2 \pi(r^{(N-1)}_{{\rm B},k}-r_{{\rm B},k})/\lambda}\right].
\end{align}
Similarly, the NLoS channel components can be modeled as
\begin{equation}
		\mathbf{h}^{H}_{{\rm NLoS},{\rm B},k}=\sqrt{\frac{N}{L_{{\rm B},k}}}\sum_{\ell=1}^{L_{{\rm B},k}}h_{{\rm B},k,\ell}\mathbf{b}^{H}(\theta_{{\rm B},k,\ell},r_{{\rm B},k,\ell}),   \vspace{-6pt}
\end{equation}
where $L_{{\rm B},k}$ denotes the number of scatterers associated with the legitimate user $k$. $h_{{\rm B},k,\ell}$ is the corresponding complex-valued channel gain of the $\ell$-th NLoS path between the origin and the legitimate user $k$. $r_{{\rm B},k,\ell}$ and $\theta_{{\rm B},k,\ell}=\cos
\phi_{{\rm B},k,\ell}$ denote the distance and spatial angle of the $\ell$-th scatterer with respect to the origin with $\phi_{{\rm B},k,\ell}$ being the physical AoD. Based on the above, the multipath near-field channel between the XL-array BS and the legitimate user $k$ is given by
\begin{equation}
	\mathbf{h}^{H}_{{\rm B},k}=\mathbf{h}^{H}_{{\rm LoS},{\rm B},k}+\mathbf{h}^{H}_{{\rm NLoS},{\rm B},k}.
\end{equation}
%  all users are located in the near-field region, the line-of-sight (LoS) channel component of each user should be modeled based on the spherical wave propagation. Let $r_{k}^{\rm B}$ and $r^{E}$ represent the distances between the XL-array center and legitimate user $k\in\mathcal{K}$ and eavesdropper, respectively.
The near-field channel between the XL-array BS and the eavesdropper can be characterized in a similar way, which is given by
\begin{equation}
	\!\!\mathbf{h}^H_{{\rm E}}= \sqrt{N}h_{{\rm E}}\mathbf{b}^H(\theta_{{\rm E}},r_{{\rm E}})  +\sqrt{\frac{N}{L_{{\rm E}}}}\sum_{\ell=1}^{L_{{\rm E}}}h_{{\rm E},\ell}\mathbf{b}^H(\theta_{{\rm E},\ell},r_{{\rm E},\ell}),
 \vspace{-3pt}
\end{equation}
where $\mathbf{b}(\theta_{{\rm E}},r_{{\rm E}})$ is the steering vector from the XL-array BS to the eavesdropper. $h_{{\rm E}}$ and $h_{{\rm E},\ell}$ denote the complex-valued channel gains. In this paper, we consider the near-field PLS in high-frequency bands, e.g., millimeter-wave (mmWave)
and terahertz (THz), for which the NLoS paths suffer from more severe path loss and thus have negligible power \cite{wu2023near,10130629}. Therefore, the channels between the XL-array BS and the legitimate user $k$ and the eavesdropper can be approximated as $\mathbf{h}^H_{{\rm B},k}\approx\mathbf{h}^{H}_{{\rm LoS},{\rm B},k}$ and $	\mathbf{h}^H_{{\rm E}}\approx \mathbf{h}^{H}_{{\rm LoS},{\rm E}}$, respectively. 
\vspace{-12pt}
\subsection{Signal Model} Let $\mathbf{x} \in \mathbb{C}^{N}$ denote the signal vector emitted by the BS to the $K$ legitimate users, consisting of $K$ information-bearing signals and AN, given by
\begin{equation}
	\mathbf{x} = \sum_{k \in \mathcal{K}} \mathbf{w}_{k} s_{k}+ \mathbf{z},
 \vspace{-3pt}
\end{equation}
where $\mathbf{w}_{k} \in \mathbb{C}^{N}$ and $s_{k} \in \mathbb{C}$ denote the beamforming vector and the information-bearing signal for $k$-th legitimate user, respectively.
In particular, to achieve secure communication, the AN vector $\mathbf{z}\in \mathbb{C}^{N}$ is injected into the desired signal that is transmitted by the BS to impair the eavesdropper.
 As such, the received signals at the legitimate user $k$ and the eavesdropper are given by
 \vspace{-0.25em}
\begin{align}
	y_{{\rm B},k} &= \mathbf{h}^{H}_{{\rm B},k}\left(\sum_{i\in\mathcal{K}} \mathbf{w}_{i}{s}_{i}+\mathbf{z} \right) + n_{{\rm B},k}, \\
	y_{\rm E} &= \mathbf{h}^{H}_{{\rm E}}\left(\sum_{i\in\mathcal{K}}\mathbf{w}_{i}{s}_{i}+\mathbf{z} \right) + n_{\rm E},
\end{align}
respectively, where $n_{{\rm B},k} \sim \mathcal{CN}(0,\sigma^2_{{\rm B},k})$  
and $n_{\rm E}\sim \mathcal{CN}(0,\sigma^2_{\rm E})$  
represent the additive white Gaussian noise (AWGN) at the legitimate user $k$ and eavesdropper, respectively. As such, the achievable rate of legitimate user $k$ is given by $R_{{\rm B},k}=\log_{2}\left(1+\gamma_{{\rm B},k}\right) $, where 
\begin{align}
\gamma_{{\rm B},k}=\frac{| \mathbf{h}_{{\rm B},k}^H\mathbf{w}_{k}| ^2}{\sum_{i\in\mathcal{K}\setminus\{k\}}| \mathbf{h}_{{\rm B},k}^H\mathbf{w}_{i}| ^2+| \mathbf{h}_{{\rm B},k}^H\mathbf{z}| ^2+\sigma^2_{{\rm B},k}}.
\end{align}
In particular, to shed new insights into the PLS in near-field communications, we consider a challenging scenario, where the eavesdropper is closer to the BS than all the legitimate users.
Furthermore, we assume that the eavesdropper is capable of canceling all multi-user interference before decoding the desired information \cite{9133130,9024490}. Accordingly, the channel capacity between the BS and the eavesdropper for wiretapping legitimate user $k$ is given by
\begin{align}
	C_{{\rm E},k}=\log_{2} \left(1+\frac{|\mathbf{h}_{\rm E}^H\mathbf{w}_{k}|^2}{|\mathbf{h}_{\rm E}^H\mathbf{z}|^2+\sigma^2_{\rm E}} \right).
\end{align}
The achievable secrecy rate of the legitimate user $k$ is given by
$R^{\rm Sec}_{k}=\left[R_{{\rm B},k}-C_{{\rm E},k}\right]^{+}$. In this paper,
we consider a general system sum secrecy rate maximization by optimizing the beamforming vectors and AN, i.e., $\{\mathbf{w}_{k}\}$ and $\mathbf{z}$. The resultant optimization problem
is formulated as
\begin{subequations}
	\begin{align}
		({\bf P1}):~~~\max_{\substack{\mathbf{w}_{k}, \mathbf{z}} }  &~~	\sum_{k \in \mathcal{K}}~\left[R_{{\rm B},k}-C_{{\rm E},k}\right]^{+}
		\\
		\text{s.t.}
		&~~ 	\sum_{k \in \mathcal{K}}\|\mathbf{w}_{k}\|^2 +\|\mathbf{z}\|^2\le P,\label{P1:pow_cons}
	\end{align}
 \end{subequations}
where \eqref{P1:pow_cons} denotes the transmit power constraint of the BS with $P$ being the maximum transmit power. 
%	It is worth noting that there is no difference between the form of problem (P1) in near-field communication and that in far-field communication.

%\section{Is Secure communication always achievable in near-field communications?}
\vspace{-6pt}
\section{Is AN Beneficial to PLS in Near-Field Communications?}\label{Sec:SC}
To shed useful insights into the near-field secure transmission design, we consider in this section a special case with one legitimate user and one eavesdropper.\footnote{The obtained results can be readily extended to more general cases involving either multiple legitimate users or multiple eavesdroppers by analyzing each legitimate user-eavesdropper pair as the considered special case.} Here, the index of legitimate user $k$ is dropped for notational convenience. 
In addition, for ease of implementation, we consider the MRT-based beamforming\footnote{While zero-forcing like beamforming can generally result in better performance, the reason for choosing the MRT-based beamforming lies in its numerous advantages in terms of analytical tractability, computational complexity, and practical implementation \cite{zhang2023swipt}.} for both the legitimate user and the eavesdropper, i.e., $\mathbf{w}=\sqrt{P_{\rm B}}\frac{\mathbf{h}_{\rm B}}{\|\mathbf{h}_{\rm B}\|}$ and $\mathbf{z}=\sqrt{P_{\rm E}}\frac{\mathbf{h}_{\rm E}}{\|\mathbf{h}_{\rm E}\|}$, to facilitate beam focusing, where $P_{\rm B}$ and $P_{\rm E}$ denote the power allocations to the legitimate user and eavesdropper, respectively. Accordingly, the secrecy rate is given by
% Specifically, for ease of analysis, we consider the pure analog beamforming only, for which it is assumed that the codebook-based beamforming vector for the legitimate user and the AN vector for the eavesdropper have been preconfigured to align with their respective channels, i.e., $\mathbf{w}=\sqrt{P_{\rm B}}\frac{\mathbf{h}_{\rm B}}{\|\mathbf{h}_{\rm B}\|}$ and $\mathbf{z}=\sqrt{P_{\rm E}}\frac{\mathbf{h}_{\rm E}}{\|\mathbf{h}_{\rm E}\|}$.
{
\begin{align}\label{Eq:sc_sec}
	&R ^{\rm Sec} =  \left[ R_{\rm B}-C_{\rm E}\right]^{+}\nn\\
	&= \left[\log_{2}\left( 1+\frac{|\mathbf{h}_{\rm B}^H\mathbf{w}|^2}{|\mathbf{h}_{\rm B}^H\mathbf{z}|^2+\sigma^2_{{\rm B}}}\right)-\log_{2}\left( 1+\frac{|\mathbf{h}_{\rm E}^H\mathbf{w}|^2}{|\mathbf{h}_{\rm E}^H\mathbf{z}|^2+\sigma^2_{\rm E}}\right)\right]^{+}\nn\\
	&\overset{(a)}{=}\left[\log_{2}\left( 1+\frac{P_{\rm B}g_{\rm B}}{P_{\rm E}g_{\rm B}|\mathbf{b}^H(\theta_{{\rm B}},r_{{\rm B}})\mathbf{b}(\theta_{{\rm E}},r_{{\rm E}})|^2+\sigma^2}\right)\nn \right.\\
	&\left. \quad \quad-\log_{2}\left( 1+\frac{P_{\rm B}g_{\rm E}|\mathbf{b}^H{(\theta_{{\rm E}},r_{{\rm E}})}\mathbf{b}(\theta_{{\rm B}},r_{{\rm B}})|^2}{P_{\rm E}g_{\rm E}+\sigma^2}\right)\right]^{+},
%	C^{\rm E}&=\log_{2}\left( 1+\frac{|(\mathbf{h}^{\rm E})^H\mathbf{w}^{\rm B}x|^2}{|(\mathbf{h}^{\rm E})^H\mathbf{z}|^2+\sigma^2}\right),\nn\\
%	&=\log_{2}\left( 1+\frac{P_{\rm B}g_{\rm E}|(\mathbf{b}^{\rm E})^H\mathbf{b}^{\rm B}|^2}{P_{\rm E}g_{\rm E}+\sigma^2}\right),
\end{align}
where $(a)$ holds by letting $\sigma^2_{{\rm B}}=\sigma^2_{{\rm E}}=\sigma^2$;\footnote{To simplify the analysis, we assume that both the legitimate user and the eavesdropper have the same noise power. While for the case where the noise power differs, this equation still holds by performing the normalization of the noise power through the respective channel gains, i.e., $g_{\rm B}$ and $g_{\rm E}$.} $g_{\rm B}=N|h_{\rm B}|^2$ and $g_{\rm E}=N|h_{\rm E}|^2$. Besides, note that dropping the operator $[\cdot]^{+}$
in \eqref{Eq:sc_sec} does not affect the
subsequent analysis. This is because the transmission would turn off 
if the achievable secrecy rate
is non-positive, and hence it is omitted in the following for notational
simplicity \cite{9024490}. To obtain a more tractable form of \eqref{Eq:sc_sec}, we first make a key definition below.
\begin{definition}
	\emph{The correlation between any two near-field steering vectors is defined as
		\begin{equation}
			\!\!\!\rho\left( \theta_{i},\theta_{j},r_{i},r_{j}\right) = |\mathbf{b}^{H}(\theta_{i},r_{i})\mathbf{b}(\theta_{j},r_{j})|, \forall i,j \in \mathcal{K}\cup\{{\rm E}\}.
	\end{equation}}
\end{definition}

\begin{remark}
	\emph{In fact, the defined correlation $\rho(\cdot)$ can be viewed as two key factors affecting the system secrecy performance from different perspectives, as described below. First, from the perspective of security provisioning, the correlation measures the amount of \emph{information leakage} from the legitimate user to the eavesdropper. On the other hand, from a communication perspective, the correlation reflects the \emph{interference intensity} among legitimate users. Besides, it is worth noting that the correlation belongs to a fixed interval, that is, $\rho \in (0,1]$ and 
 has a symmetry property, i.e., $	\rho(\theta_{i},r_{i},\theta_{j},r_{j})=	\rho(\theta_{j},r_{j},\theta_{i},r_{i})$. Accordingly, the secrecy rate in \eqref{Eq:sc_sec} can be re-expressed as
	\begin{align}\label{Eq1:secrecyrate_AN}
			R ^{\rm Sec}&=\log_{2}\left( 1+\frac{P_{\rm B}g_{\rm B}}{P_{\rm E}g_{\rm B}\rho^2+\sigma^2}\right) -\log_{2}\left( 1+\frac{P_{\rm B}g_{\rm E}\rho^2}{P_{\rm E}g_{\rm E}+\sigma^2}\right)\nn\\
			%		&=\log_{2}\left( \frac{P_{\rm E}g_{\rm B}\rho+\sigma^2+P_{\rm B}g_{\rm B}}{P_{\rm E}g_{\rm B}\rho+\sigma^2}\cdot \frac{P_{\rm E}g_{\rm E}+\sigma^2}{P_{\rm E}g_{\rm E}+\sigma^2+P_{\rm B}g_{\rm E}\rho}\right) \nn\\
			% &=\log_{2}\left( \frac{P_{\rm E}^2g_{\rm B}g_{\rm E}\rho^2+P_{\rm E}g_{\rm B}\rho^2\sigma^2+P_{\rm E}g_{\rm E}\sigma^2+\sigma^4+P_{\rm B}P_{\rm E}g_{\rm B}g_{\rm E}+P_{\rm B}g_{\rm B}\sigma^2}{P_{\rm E}^2g_{\rm B}g_{\rm E}\rho^2+P_{\rm E}g_{\rm B}\rho^2\sigma^2+P_{\rm E}g_{\rm E}\sigma^2+\sigma^4+P_{\rm B}P_{\rm E}g_{\rm B}g_{\rm E}\rho^4+P_{\rm B}g_{\rm E}\rho^2\sigma^2}\right)
			&\triangleq\log_{2}\left( \frac{A+B}{A+C}\right),
		\end{align}
  where
	\begin{equation}\label{Eq:coef}
		\begin{cases}
			A=P_{\rm E}^2g_{\rm B}g_{\rm E}\rho^2+P_{\rm E}g_{\rm B}\rho^2\sigma^2+P_{\rm E}g_{\rm E}\sigma^2+\sigma^4, \\
			B=P_{\rm B}P_{\rm E}g_{\rm B}g_{\rm E}+P_{\rm B}g_{\rm B}\sigma^2,\\
			C=P_{\rm B}P_{\rm E}g_{\rm B}g_{\rm E}\rho^4+P_{\rm B}g_{\rm E}\rho^2\sigma^2.
		\end{cases}
		\end{equation}}
\end{remark}

%\begin{definition}
%	\emph{The interference leakage between any two legitimate users $p, q\in\mathcal{K}$ characterizes its (normalized) beam power in between, defined as
%		\begin{equation}
%			\rho\left( \theta_{{\rm B},p},\theta_{{\rm B},q},r_{{\rm B},p},r_{{\rm B},q}\right) = |\mathbf{b}^{H}(\theta^{\rm B}_{p},r^{\rm B}_{p})\mathbf{b}(\theta^{\rm B}_{q},r^{\rm B}_{q})|.
%	\end{equation}}
%\end{definition}
\subsection{Near-field Secure Transmission without AN}
First, we consider the case without AN and derive the secure condition directly.
\begin{lemma}\label{Pro:secure_reg}
	\emph{The near-field secure transmission without AN is attainable only when the following condition holds
		\begin{equation}\label{Eq:secure_condi}
			g_{\rm B}-g_{\rm E}\rho^2 \ge 0.
	\end{equation}}
\end{lemma}

\begin{proof}
	By removing the terms related to AN in \eqref{Eq1:secrecyrate_AN}, we have
	\begin{align}\label{Eq1:con_secrecyrate}
		R^{\rm Sec}_{\rm w/o} 
    =\log_{2}\left( 1+\frac{P_{\rm B}\left(g_{\rm B}-g_{\rm E}\rho^2 \right) }{P_{\rm B}g_{\rm E}\rho^2+\sigma^2}\right).
	\end{align}
	The system is secure only when the condition holds, i.e., $g_{\rm B}-g_{\rm E}\rho^2 \ge 0$. This completes the proof.
\end{proof}

Lemma \ref{Pro:secure_reg} provides the secure condition for the case where AN is not adopted and reveals three key factors that affect the
near-field secure transmission: the channel gains of the legitimate user and the eavesdropper, as well as the correlation between them. Notably, the secure condition is of the differential form and can, in fact, be considered as the \emph{effective} quantity of information received by the legitimate user. Specifically, the first term indicates the quantity of transmitted information, whereas the second item signifies the information leaked to the eavesdropper.
However, it is intractable to analyze the effects of these factors on the secure condition because the correlation generally decreases when the legitimate user is farther away from the eavesdropper, and vice versa.  Furthermore, it can be observed that when the legitimate user and eavesdropper reside in the same spatial angle, the distance interval for satisfying the secure condition is very narrow. Outside of this interval, the near-field secure communication is not achievable. For example, when $N=256$, $f=100$ GHz, $\theta_{\rm E}=\theta_{\rm B}=0$, and $r_{\rm E}=3$ m, the secure region for the legitimate user is approximately $[3.5,8.2]$ m, which is around $5\%$ portion of the near-field region. This leads to an intriguing question: \emph{Is AN advantageous to systems where the secure condition cannot be met?}  To answer this question, the next subsection will delve into the near-field secure transmission with AN taken into account.
\vspace{-12pt}
\subsection{Near-field Secure Transmission with AN}
We now consider the case where the AN is adopted for security provisioning. The resultant condition for guaranteeing security is presented as follows.
%\begin{definition}[Information leakage]
%	\emph{The information leakage from any legitimate user $k\in\mathcal{K}$ characterizes its (normalized) beam power at the eavesdropper, which is given by
%	\begin{equation}
%			\rho\left( \theta^{\rm B}_{k},\theta^{\rm E},r^{\rm B}_{k},r^{\rm E}\right) = \left| \mathbf{b}^{H}(\theta^{\rm B}_{k},r^{\rm B}_{k})\mathbf{b}(\theta^{\rm E},r^{\rm E})\right| .
%		\end{equation} }
%\end{definition}

%\begin{definition}[Secure region]
%	\emph{The secure region characterizes the distance interval $r^{\rm B}\in \left[ r^{\rm B}_{\rm min}, r^{\rm B}_{\rm max}\right] $, within which $R \ge 0$ always holds, which is defined as }
%\end{definition}

%\begin{remark}
%	\emph{It is clearly observed from \eqref{Eq1:con_secrecyrate} that the \emph{secrecy region} is significantly restricted under the considered scenario. Outside of this region, the system becomes insecure. To resolve this issue, AN-based beamforming is incorporated to improve the secrecy performance. However, it remains unknown whether AN is helpful to the near-field PLS, especially under the considered setting.  }
%\end{remark}

\begin{lemma}\label{Pro:sec_condi}
	\emph{When the secure condition \eqref{Eq:secure_condi} does not hold, i.e., $g_{\rm B}-g_{\rm E}\rho^2 < 0$, the system becomes insecure. In such case, if
     \begin{equation}\label{Eq:secure_conddix}
			g_{\rm B}-g_{\rm E}\rho^2 \ge -\frac{Pg_{\rm B}g_{\rm E}(1-\rho^4)}{\sigma^2},
	\end{equation}
 the system can regain security by allocating at least $P_{{\rm E, min}}$ power to AN. In particular,
\begin{equation}\label{Eq:req_sec_pow}
    P_{{\rm E, min}}=\frac{(g_{\rm E}\rho^2-g_{\rm B})\sigma^2}{g_{\rm B}g_{\rm E}(1-\rho^4)}<P.
\end{equation}
}
\end{lemma}
\begin{proof}
    Please refer to Appendix \ref{App1}.
\end{proof}
		\begin{remark}[How does AN make an impact?]
			\emph{Regarding the derived condition in \eqref{Eq:secure_conddix}, we can draw several important insights: 1) The AN is crucial for near-field security provisioning, transforming insecure systems into secure ones. Moreover, it is observed from \eqref{Eq:req_sec_pow} that the power required for AN is irrespective of the total transmit power, but depends on the system setting, i.e., the location information of legitimate user and eavesdropper. 2) The condition for regaining security is jointly determined by the transmit power budget, channel gains of the legitimate user and eavesdropper, and the correlation coefficient. When the location information
of the legitimate user and the eavesdropper (i.e., angle
and distance) is known $a$ $priori$, these factors are fixed. 3) It is worth noting that once the minimum power for secrecy provisioning is provided, the secure condition is independent of the power allocation design. This suggests that the secure condition serves as a boundary that is inherently embedded in the considered system. 4) From Lemma \ref{Pro:sec_condi}, we can conclude that there must exist the maximum secrecy rate of the considered system as long as the secure condition is satisfied properly. 5) The power allocated to AN for secure transmission can actually serve as an indicator of how secure a legitimate user is relative to the eavesdropper, hence acting as a crucial component of the proposed low-complexity approach in Section \ref{Se:Low-com}.}
		\end{remark}

		The secrecy rate $R^{\rm Sec}$ in \eqref{Eq1:secrecyrate_AN} is still in a complicated form, which, in general, is challenging to characterize its maximum value. To tackle this issue, we optimize the power allocation to achieve the maximum secrecy rate.

%		where we have 
%		\begin{align}
%			\tilde{B}&=-P_{\rm E}^2g_{\rm B}g_{\rm E}+P_{\rm E}(Pg_{\rm B}g_{\rm E}-g_{\rm B}\sigma^2)+Pg_{\rm B}\sigma^2, \nn\\ 
%			\tilde{C}&=-P_{\rm E}^2g_{\rm B}g_{\rm E}\rho^4+P_{\rm E}(Pg_{\rm B}g_{\rm E}\rho^4-g_{\rm E}\rho^2\sigma^2)+Pg_{\rm E}\rho^2\sigma^2.
%		\end{align}
%		
	
\begin{lemma}\label{Le:opt_pow}
	\emph{The optimal power allocation to achieve the maximum secrecy rate is given as follows. 
		\begin{itemize}
  			\item If the following condition is satisfied
   \begin{equation}\label{Eq:sec_max}
       g_{\rm B}-g_{\rm E}\rho^2< \frac{P^2 g_{\rm B}g_{\rm E}(g_{\rm E}-g_{\rm B}\rho^2)\rho^2+P(g_{\rm E}^2-g_{\rm B}^2)\rho^2\sigma^2}{\sigma^4},
   \end{equation}
            the achievable secrecy rate first increases and then decreases, with the maximum attained with the following power allocation
			\begin{equation}\label{Eq:pow_max}
				\begin{cases}
					P_{\rm E}= -\frac{\Upsilon_2}{{2\Upsilon_1}}-\frac{\sqrt{{\Upsilon_2^2}-4{\Upsilon_1}\Upsilon_3}}{2\Upsilon_1},\\
					P_{\rm B}=P-P_{\rm E},
				\end{cases}
			\end{equation}
   	where
		\begin{equation}\label{Eq:Deriv_Coe}
			\begin{cases}
				\Upsilon_1  =P   g_{\rm B}^2   g_{\rm E}^2 \rho^2(\rho^4-1)+   g_{\rm B}^2 g_{\rm E} \rho^4 (\rho^2-1)\sigma^2\\
				~~~~~~~+  g_{\rm B}g_{\rm E}^2(\rho^2-1)\sigma^2, \\
				\Upsilon_2  =2 P  g_{\rm B} g_{\rm E}(g_{\rm B}+g_{\rm E})(\rho^2-1)\rho^2  \sigma^2\\
                ~~~~~~~+2g_{\rm B}g_{\rm E}(\rho^4-1)\sigma^4,\\
				\Upsilon_3  =P^2  g_{\rm B}  g_{\rm E} (g_{\rm E}-g_{\rm B}\rho^2)\rho^2\sigma^2 +P  (g_{\rm E}^2 -g_{\rm B}^2)\rho^2\sigma^4\\
				~~~~~~~+(g_{\rm E}\rho^2 -g_{\rm B})\sigma^6 .
			\end{cases}
	\end{equation}
			\item Otherwise,
% \begin{equation}\label{Eq:sec_mon}
%          g_{\rm B}-g_{\rm E}\rho^2 \ge  \frac{P^2 g_{\rm B}g_{\rm E}(g_{\rm E}-g_{\rm B}\rho^2)\rho^2+P(g_{\rm E}^2-g_{\rm B}^2)\rho^2\sigma^2}{\sigma^4},
%     \end{equation}
          the secrecy rate is monotonically decreasing with $P_{\rm E}$, and the maximum is achieved without allocating power to AN, i.e.,
		\begin{equation}\label{Eq:pow_mon}
			\begin{cases}
				P_{\rm E}= 0,\\
				P_{\rm B}=P.
			\end{cases}
		\end{equation}
		\end{itemize}}
\end{lemma}

\begin{proof}
		Please refer to Appendix \ref{App2}.
\end{proof}

 \begin{figure*}[t]
	\centering
	\subfigure[\textbf{Case 1:} secure condition \eqref{Eq:secure_condi} does not hold but \eqref{Eq:secure_conddix} is met.]{\includegraphics[width=0.32\textwidth]{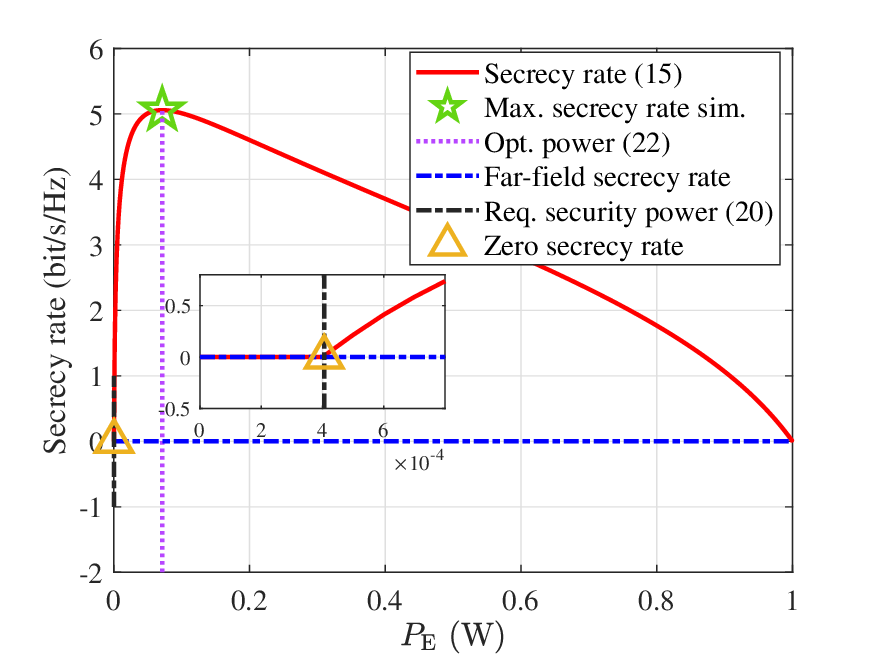}\label{fig:insecure}}
	\subfigure[\textbf{Case 2:} secure condition holds and \eqref{Eq:sec_max} is satisfied. ]{\includegraphics[width=0.32\textwidth]{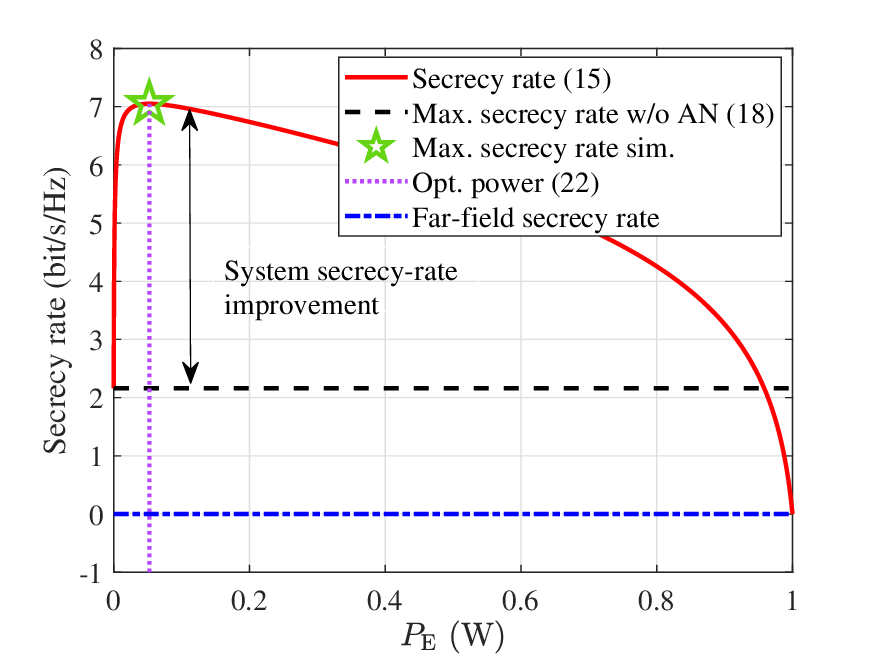}\label{fig:maxsecure}}
	\subfigure[\textbf{Case 3:} condition in \eqref{Eq:sec_max} does  not hold.]{\includegraphics[width=0.32\textwidth]{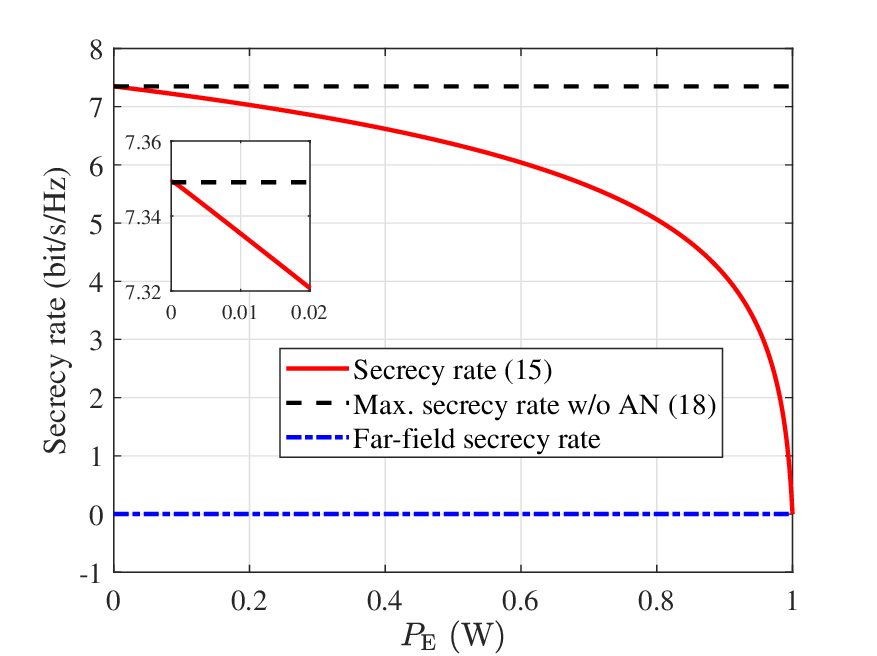}\label{fig:monsecure}}
	\caption{Three cases of secrecy rate versus the power allocated to AN.}\label{Fig:AN_effect}
 \vspace{-22pt}
\end{figure*}

Lemma \ref{Le:opt_pow} reveals an interesting result related to the effect of AN. Note that for the case without AN taken into account, it can be easily verified by Lemma \ref{Pro:secure_reg} that when the secure condition is satisfied, the achievable secrecy rate increases monotonically with the power allocated to the legitimate user. However, for the case where AN is considered, the bound on which the secrecy rate changes from increasing first and then decreasing to monotonically decreasing is shown to be strictly greater than zero, i.e., $\frac{P^2 g_{\rm B}g_{\rm E}(g_{\rm E}-g_{\rm B}\rho^2)\rho^2+P(g_{\rm E}^2-g_{\rm B}^2)\rho^2\sigma^2}{\sigma^4} > 0$. This suggests that even if the secure condition holds, i.e., $g_{\rm B}-g_{\rm E}\rho^2 \ge 0$, the secrecy rate will first increase and then decrease, indicating the ability of AN in enhancing the secrecy rate compared to the case without AN.

%\begin{figure}[t]
% 	\centering
%	\includegraphics[width=0.48\textwidth]{case1.eps}\label{fig:SC}
% 	\caption{x}
%\end{figure} 

\begin{example}\label{Ex:AN}
	\emph{By means of Lemmas \ref{Pro:sec_condi} and \ref{Le:opt_pow}, we show the necessity of AN not only in the security provisioning but also in the secrecy rate enhancement.  
		%convert the insecure circumstances into secure situations is unveiled. 
		We in the following present a concrete example to illustrate the effect of AN in detail. Specifically, we consider a near-field PLS system with $N=256$, $f=100$ GHz and $P = 30$ dBm, in which the legitimate user and eavesdropper are located at polar coordinates $(0,0.05r_{\rm Ray})$ and $(0,0.3r_{\rm Ray})$, respectively. We depict the achievable secrecy rate versus the power allocation to the eavesdropper in Fig. \ref{Fig:AN_effect} to show the effect of AN on the achievable secrecy rate. Based on Lemmas \ref{Pro:sec_condi} and \ref{Le:opt_pow}, the effect of AN on the secrecy rate can be divided into three cases.
  \begin{itemize}
        \item \textbf{Case 1:} When the system is insecure, i.e., $g_{\rm B}-g_{\rm E}\rho^2 < 0$, but the condition \eqref{Eq:secure_conddix} is met, it is shown in Fig. \ref{fig:insecure} that the secrecy rate first increases and then decreases with $P_{\rm E}$, starting from the secrecy rate less than zero. Besides, it is observed that the insecure system regains its security using only a marginal amount of power, which is consistent with the required power derived in \eqref{Eq:req_sec_pow}. Moreover, we observe that the maximum secrecy rate exists (marked by a pentagram), and the corresponding power allocation matches well with the derived optimal power allocation in \eqref{Eq:pow_max}.
        \item \textbf{Case 2:} For this case where the system is secure and the condition \eqref{Eq:sec_max} holds, several important observations are made as follows. Interestingly, we see that the trend in the secrecy rate is similar to that in Case 1, yet it is entirely distinct from the trend in the scenario without AN as described in Lemma \ref{Pro:secure_reg}. Besides, one can observe that even if the system is secure, allocating a moderate amount of power to the AN can significantly enhance the system secrecy rate.
        Additionally, it is worth noting that in most cases, the power allocation to AN is beneficial to the improvement of the secrecy rate. This is due to the fact that the AN can fully take advantage of the unique beam focusing characteristic, thereby significantly impairing the eavesdropper's performance.
        \item \textbf{Case 3:} For this case, the condition in \eqref{Eq:sec_max} does not hold, the secrecy rate monotonically decreases with the power allocated to the eavesdropper. In other words, the information leaked to the eavesdropper is very limited, and thus AN is no longer needed. 
  \end{itemize} }
\end{example}

   \vspace{-3pt}
In what follows, we investigate the effects of angle and distance on the achievable secrecy rate. First, we present a useful lemma below.
\begin{lemma}\label{Le:corre}
	\emph{The correlation of two near-field steering vectors $\rho(\theta_{i},r_{i},\theta_{j},r_{j})$ can be approximated as
		\begin{align}\label{Eq:correlation_fre}	\!\!\!\!\rho(\theta_{i},r_{i},\theta_{j},r_{j})&=G(\beta_1,\beta_2)\approx\left|\!\frac{\tilde{C}(\beta_1,\beta_2)\!+\!\jmath\tilde{S}(\beta_1,\beta_2) }{2\beta_2}\!\right|\!,\!
		\end{align}
		where $\tilde{C}(\beta_1,\beta_2)={C}(\beta_1+\beta_2)-{C}(\beta_1-\beta_2)$ and $\tilde{S}(\beta_1,\beta_2) ={S}(\beta_1+\beta_2)-{S}(\beta_1-\beta_2)$ with ${C}(\beta)=\int_{0}^{\beta}\cos(\frac{\pi}{2}t^2)dt$ and ${S}(\beta)=\int_{0}^{\beta}\sin(\frac{\pi}{2}t^2)dt$ being the Fresnel integrals; and  
		\begin{equation}\label{beta1beta2}
			\beta_1=\frac{(\theta_{j}-\theta_{i})}{\sqrt{d\left|\frac{1-\theta_{i}^2}{r_i}- \frac{1-\theta_{j}^2}{r_j}\right| }},~~\beta_2=\frac{N}{2}\sqrt{d\left|\frac{1-\theta_{i}^2}{r_i}- \frac{1-\theta_{j}^2}{r_j}\right| }.
	\end{equation} }
\end{lemma}
\begin{proof}
	The proof is similar to that in [Lemma 1,\cite{10195974}] and
	hence is omitted for brevity.
\end{proof}
\begin{proposition}\label{apppro1}
   \emph{The secrecy rate is monotonically decreasing with respect to the correlation $\rho$ between the legitimate user and eavesdropper.
   Moreover, combined with Lemma \ref{Le:corre}, we can draw the following insights:
   \begin{itemize}
       \item The secrecy rate is symmetric in terms of the angle difference between the legitimate user and eavesdropper, i.e., $\theta_{\rm B}-\theta_{\rm E}$. Besides, the secrecy rate reaches its minimum value when the angle difference is zero and, in general, decreases with the angle difference from both sides. This is because $\rho$ in \eqref{Eq:correlation_fre} is symmetric regarding $\beta_1$  in \eqref{beta1beta2}, the maximum $\rho$ is obtained at $\theta_{\rm B}=\theta_{\rm E}$.
       \item The effect of distance on the secrecy rate is considerably complicated and thus difficult to analyze. Specifically, consider a typical case where $\theta_{\rm B}=\theta_{\rm E}$. The envelope of $\rho$ is shown to decrease as $\beta_2$ in \eqref{beta1beta2} increases \cite{zhang2023mixed} and thus decreases with the distance difference. However, the distance variation also affects the channel gains and the secrecy rate. The effect will be evaluated numerically later. 
   \end{itemize} }
\end{proposition}
\begin{proof}
   The proof is similar to that of Lemma \ref{Le:opt_pow}, and thus is omitted here for brevity.
\end{proof}
\begin{example}[How do the angle and distance affect?] \emph{To comprehensively demonstrate the effect of AN on near-field PLS, we further provide a concrete example in Fig. \ref{Fig:NF_angle_dist}, illustrating the effects of the spatial angle and distance of the legitimate user on the achievable secrecy rate. Specifically, the system setting is depicted in Fig. \ref{fig:NF_angle_dist_SM}, and the secrecy rates versus the legitimate user angle and distance are plotted in Figs. \ref{fig:NF_angle} and \ref{fig:NF_dist}, respectively, with the derived optimal power allocation in Lemma \ref{Le:opt_pow}. First, it is observed from Fig. \ref{fig:NF_angle} that the secrecy rate is symmetric with respect to the angle difference. The secrecy rate tends to degrade with the decrease of angle difference and reaches the minimum value when the angle difference is zero, which is in accordance with Proposition \ref{apppro1}. An interesting observation is that the secrecy rate of the case aided by AN is smooth, while the secrecy rate of the case without AN appears to fluctuate. Moreover, as the angle difference increases, the fluctuations decrease. This is because the injection of AN can efficiently suppress the information leakage. In contrast, in the case without AN, the information is leaked to the eavesdropper to some extent due to the fluctuating correlation between them. Moreover, we see from Fig. \ref{fig:NF_dist} that AN is capable of guaranteeing the whole-distance secure transmission, while the scheme without AN fails to provide security gain. Additionally, it is observed that the secrecy rate first increases and then decreases with the legitimate user distance, rendering its non-linearity in the distance domain. This is consistent with the results in Proposition \ref{apppro1}.}
\end{example}

\vspace{-9pt}
\section{Low-Complexity Approach for Problem (P1)}\label{Se:Low-com}
The above analytical results are derived under the premise of analog beam focusing-based PLS with a single legitimate user. In the following, we investigate the more general case with multiple legitimate users. We shall first present the conventional approach and then propose a novel low-complexity approach for near-field secure transmission. 

\begin{figure*}[t]
	\centering
	\subfigure[System setting.]{\includegraphics[width=3.8cm]{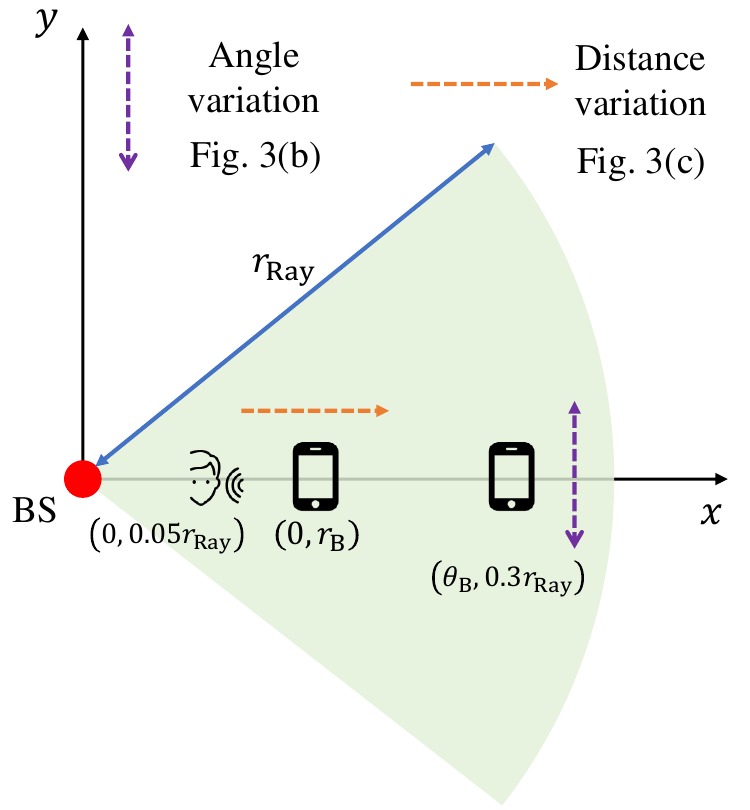}\label{fig:NF_angle_dist_SM}}	
	\subfigure[Secrecy rate vs. legitimate user angle.]{\includegraphics[width=5.35cm]{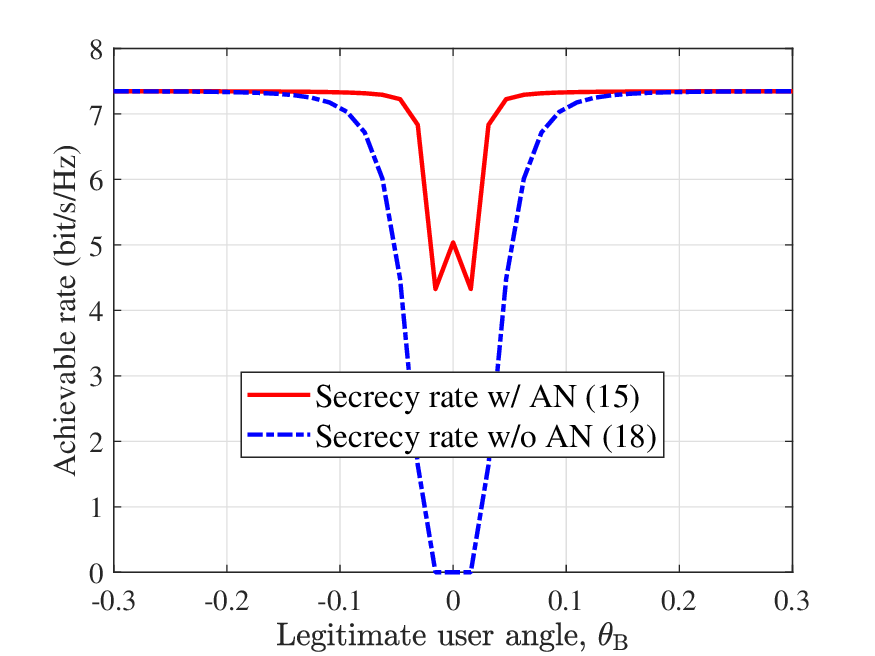}\label{fig:NF_angle}}
	\subfigure[Secrecy rate vs. legitimate user distance.]	{\includegraphics[width=5.35cm]{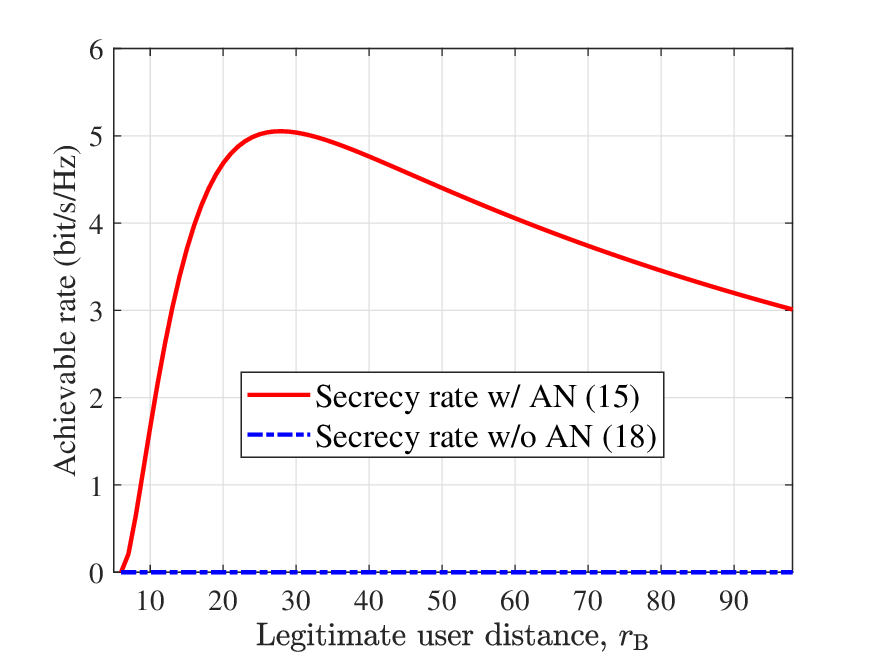}\label{fig:NF_dist}}
	\caption{: Illustration of the secrecy rate versus the angle and distance of the legitimate user.}\label{Fig:NF_angle_dist}
 \vspace{-0.5cm}
\end{figure*}
 \vspace{-8pt}
\subsection{Conventional Approach}\label{Subsec:con}
Problem (P1) is shown to be non-convex; thus it is hard to obtain the corresponding optimal solution. To this end, an efficient approach is proposed to obtain a high-quality solution, which is elaborated as:
First, by defining $\mathbf{H}_{{\rm B},k}=\mathbf{h}_{{\rm B},k}\mathbf{h}_{{\rm B},k}^H$, $\mathbf{W}_{k}=\mathbf{w}_{k}\mathbf{w}_{k}^H$, $\mathbf{Z}=\mathbf{z}\mathbf{z}^H$, and $\mathbf{H}_{\rm E}=\mathbf{h}_{{\rm E}}\mathbf{h}_{{\rm E}}^H$, we can equivalently recast problem (P1) as:
\begin{subequations}
	\begin{align}
		({\bf P2}):~~~\min_{\substack{\mathbf{W}_{k}, \mathbf{Z}} }  &~~	f_{\rm obj}=N_1+N_2-D_1-D_2 \label{P2:obj}
		\\
		\text{s.t.}
		&~~ 	\sum_{k \in \mathcal{K}}\mathrm{Tr}(\mathbf{W}_{k}) +\mathrm{Tr}(\mathbf{Z})\le P,\label{P2:pow_cons} \\ 
        &~~\mathbf{W}_{k}\succeq 0,  \forall k,~~\mathbf{Z} \succeq 0, \\
        &~~\mathrm{Rank}(\mathbf{W}_{k}) \le 1, ~~ \forall k,\label{P2:rank_cons} \\
        &~~\mathrm{Rank}(\mathbf{Z}) \le 1,\label{P2:rank_cons1} 
	\end{align}
 \end{subequations}
where 
\begin{align}\label{Eq:Deriv_Coe}
        N_1  &=-\sum_{k\in\mathcal{K}}\log_2\left(\sum_{i\in\mathcal{K}}\mathrm{Tr}(\mathbf{H}_{{\rm B},k}\mathbf{W}_{i})+\mathrm{Tr} (\mathbf{H}_{{\rm B},k}\mathbf{Z})+\sigma^2_{{\rm B},k}\right), \nn\\
        N_2 &= -K\log_2\left(\mathrm{Tr}(\mathbf{H}_{\rm E}\mathbf{Z})+\sigma^2_{\rm E}\right),\nn\\
        D_1 & = -\sum_{k\in\mathcal{K}}\log_2\!\left(\!\sum_{i\in\mathcal{K}\setminus\{k\}}\!\!\!\mathrm{Tr}(\mathbf{H}_{{\rm B},k}\mathbf{W}_{i})\!+\!\mathrm{Tr} (\mathbf{H}_{{\rm B},k}\mathbf{Z})\!+\!\sigma^2_{{\rm B},k}\!\right),\nn\\
        D_2 &= -\sum_{k\in\mathcal{K}}\log_2\left(\mathrm{Tr}(\mathbf{H}_{\rm E}\mathbf{W}_{k})+\mathrm{Tr}(\mathbf{H}_{\rm E}\mathbf{Z})+\sigma^2_{\rm E}\right).
\end{align}
The main difficulties in solving the problem (P2) arise from the non-convex objective function \eqref{P2:obj} and rank constraints \eqref{P2:rank_cons} and \eqref{P2:rank_cons1}, the solutions to which are described as follows. First, the rank-one constraints can be removed by invoking the SDR technique. While the non-convex objective function is replaced by the convex upper bound derived from the SCA method. More specifically, take $D_1$ as an example. Its global underestimators 
at any feasible point $(\mathbf{W}^{(t)},\mathbf{Z}^{(t)})$ can be constructed by deriving the corresponding first-order Taylor approximation, given by
\begin{align}
    D_1(\mathbf{W},\mathbf{Z})&\ge D_1(\mathbf{W}^{(t)},\mathbf{Z}^{(t)})\nn\\
    &+\mathrm{Tr}(\nabla^H_{\mathbf{W}}D_1(\mathbf{W}^{(t)},\mathbf{Z}^{(t)})(\mathbf{W}-\mathbf{W}^{(t)}))\nn\\
    &+\mathrm{Tr}(\nabla^H_{\mathbf{Z}}D_1(\mathbf{W}^{(t)},\mathbf{Z}^{(t)})(\mathbf{Z}-\mathbf{Z}^{(t)})),
\end{align}
% \begin{align}
%     D_2(\mathbf{W},\mathbf{Z})&\ge D_2(\mathbf{W}^{(t)},\mathbf{Z}^{(t)})\nn\\
%     &+\mathrm{Tr}(\nabla^H_{\mathbf{W}}D_2(\mathbf{W}^{(t)},\mathbf{Z}^{(t)})(\mathbf{W}-\mathbf{W}^{(t)}))\nn\\
%     &+\mathrm{Tr}(\nabla^H_{\mathbf{Z}}D_2(\mathbf{W}^{(t)},\mathbf{Z}^{(t)})(\mathbf{Z}-\mathbf{Z}^{(t)})),
% \end{align}
where the superscript $(t)$ denotes the iteration index of the associated optimization variable. As such, problem (P2) now becomes convex with respect to all optimization variables and thus can be efficiently solved by standard convex solvers, e.g., CVX. However, the conventional approach entails extremely high computational complexity, i.e., $\mathcal{O}((K+1)N^{6.5})$, especially in the considered XL-array systems, hence hindering its practical implementation. To tackle the above issue, we propose a novel low-complexity approach.

 \vspace{-9pt}
\subsection{Low-Complexity Approach}
Recall that the secrecy performance of the considered system is significantly affected by two aspects: \emph{information leakage} and \emph{inter-user interference}. These factors, however, are considered to be in conflict with each other due to the intrinsic trade-off between maximizing the sum-rate and minimizing the information leakage. In order to strike a good balance, the low-complexity approach should be properly designed to make the interference between legitimate users as small as possible, while at the same time reducing the information leakage as much as possible. To account for these effects, the proposed low-complexity approach is composed of three components: namely \emph{interference domain characterization}, \emph{interference set construction}, and \emph{beamformer determination}, the detailed procedures of which are elucidated below. 

%
%To be more specific, the fundamental idea of the low-complexity algorithm comes from two aspects: 1) \emph{secure condition}, which measures the ``amount'' of information leakage from the legitimate user to the eavesdropper; 2) \emph{inter-user interference}, which accounts for the inter-user interference among legitimate users. In order to achieve secure communication, it is expected to have as little interference among legitimate users as possible, while at the same time keeping the information leakage as small as possible. In practice, however, these two aspects , thereby complicating the low-complexity scheme design.

%strike a perfect balance

\subsubsection{Interference Domain Characterization}
To characterize the interference domain, we first provide the following key propositions, namely beam width and beam depth.

\begin{proposition}[$\varphi$-{\rm dB} Beam width]\label{Pro:BW}
	\emph{For a legitimate user $k\in \mathcal{K}$, residing at $(\theta_{{\rm B},k},r_{{\rm B},k})$, the $\varphi$-{\rm dB} beam width characterizes the spatial angular width on the \emph{distance ring}, i.e.,  $\{(\theta,r)|\frac{1-\theta^2}{r}=\frac{1-\theta_{{\rm B},k}^2}{r_{{\rm B},k}}\}$ \cite{9693928}, for which the normalized beam gain at an arbitrary observed angle-distance pair $(\theta,r)$ is above $10^{-\frac{\varphi}{20}}\in(0,1]$,
		%satisfies $\eta(\theta,\theta,r,r) \le \kappa\max_{\theta'\in [-1,1]}\{\eta(\theta,\theta',r,r)\}$. 
	i.e.,
		\begin{align}
			\text{BW}_{{\rm \varphi dB},k} &= 
			%|\theta_{\rm right}-\theta_{\rm left}|= 
			\frac{2 \bar{\beta}_{\varphi}\lambda r_{{\rm B},k}}{Nd},
	\end{align}
where $\sinc(x)=\sin(\pi x)/(\pi x)$ and  $\sinc(\bar{\beta}_{\rm \varphi})={10^{-\frac{\varphi}{20}}}$.}
\end{proposition}
\begin{proof}
	The proof is similar to that of \cite{10273772}, and thus omitted for brevity.
\end{proof}
Next, the $\varphi$-{\rm dB} beam depth is characterized as follows.

\begin{proposition}[$\varphi$-{\rm dB} beam depth]\label{Pro:BD}
	\emph{For a legitimate user located at $(\theta_{{\rm B},k},r_{{\rm B},k})$, the $\varphi$-{\rm dB} beam depth characterizes the distance interval along the spatial angle $\theta_{{\rm B},k}$, within which the normalized beam gain is larger than $10^{-\frac{\varphi}{20}}\in(0,1]$.
	Mathematically, $\text{BD}_{{\rm \varphi dB},k}$ is given by
\begin{equation}\label{Eq:BD}
	\text{BD}_{{\rm \varphi dB},k} \triangleq r_{{\rm R},k}-r_{{\rm L},k} =\begin{cases}
		\frac{2r_{{\rm B},k}^2r_{\rm BD}}{r^2_{\rm BD}-r_{{\rm B},k}^2}, & r_{{\rm B},k} < r_{\rm BD},\\
		\infty,  &r_{{\rm B},k}  \ge r_{\rm BD},
	\end{cases}
\end{equation}
	where $r_{{\rm L},k}=\frac{r_{{\rm B},k}r_{\rm BD}}{r_{\rm BD}+r_{{\rm B},k}}$ and $r_{{\rm R},k}=\frac{r_{{\rm B},k}r_{\rm BD}}{r_{\rm BD}-r_{{\rm B},k}}$ denote the left and right boundaries of the beam-depth interval, respectively. $r_{\rm BD}=\frac{N^2d^2(1-\theta_{{\rm B},k}^2)}{2\lambda\beta_{\varphi}^2}$ and $G(\beta_{\varphi})={10^{-\frac{\varphi}{20}}}$.}
	\end{proposition}
%		\begin{equation}
%			\text{BD}_{\rm 3dB} = \frac{r^{\rm B}	r_{\rm DF}}{	r_{\rm DF}+r^{\rm B}}-\frac{r^{\rm B}	r_{\rm DF}}{	r_{\rm DF}-r^{\rm B}} = \frac{2(r^{\rm B})^2r_{\rm DF}}{r^2_{\rm DF}-(r^{\rm B})^2}
%		\end{equation}
%		If $r^{\rm B} \ge r_{\rm DF} $, we have
%		\begin{equation}
%			\left[\frac{r_{k}	r_{\rm bdb}}{	r_{\rm DF}+r^{\rm B}}, \infty \right] 
%	\end{equation}

\begin{proof}
	The proof is similar to that of \cite{zhao2024performance}, and thus omitted for brevity.
\end{proof}
\begin{remark}[What affects the beam depth?] 
	\emph{Proposition \ref{Pro:BD} indicates that the beam depth is considerably affected by the boundary $r_{\rm BD}$, which is determined by the antenna size and the spatial angle.  Note that $r_{\rm BD}$ is irrespective of the user distance, thus serving as a boundary when the spatial angle is fixed.  Moreover, it can be verified that the beam depth is monotonically increasing as the user distance approaches the boundary, i.e., $r_{{\rm B},k}\to r_{\rm BD}$. 
	}
\end{remark}

% in Propositions \ref{Pro:BW} and \ref{Pro:BD}

The $\varphi$-{\rm dB} beam width and beam depth in Propositions \ref{Pro:BW} and \ref{Pro:BD} jointly characterize the surrounding area of a typical user within which the beam gain of this user is less than a predefined $\varphi$-{\rm dB} threshold. In other words, if a beam steered towards another user falls into that user's surrounding area, it will cause no less than $\varphi$-{\rm dB} interference to that user. Therefore, the surrounding area can be regarded as the interference domain of a typical user, which is mathematically characterized as follows. 

\begin{proposition}[$\varphi$-{\rm dB} interference domain]
	\emph{For a legitimate user $k \in \mathcal{K}$, its $\varphi$-{\rm dB} interference domain, denoted as $\mathcal{A}_{k}$, is approximated as a rectangular region formed by the corresponding beam width and beam depth, i.e., $\text{BW}_{{\rm \varphi dB},k}$ and $\text{BD}_{{\rm \varphi dB},k}$, given by
	\begin{align}
		\mathcal{A}_{k} =\Big\{(x,y)|&\sqrt{[(x-r_{{\rm B},k}({1-\theta^2_{{\rm B},k}})^{\frac{1}{2}})^2+(y-r_{{\rm B},k}\theta_{{\rm B},k})^2]} \nn\\
		&\le\sqrt{[(\text{BD}_{{\rm \varphi dB},k}/2)^2+(\text{BW}_{{\rm \varphi dB},k}/2)^2]}, \nn\\ &|x-r_{{\rm B},k}({1-\theta^2_{{\rm B},k}})^{\frac{1}{2}}|\le \text{BD}_{{\rm \varphi dB},k}/2,\nn\\
		&|y-r_{{\rm B},k}\theta_{{\rm B},k}|\le \text{BW}_{{\rm \varphi dB},k}/2\Big\}.
		\end{align}}
\end{proposition}

%\begin{proposition}[$\varphi$-{\rm dB} Domain]
%	\emph{the $\varphi$-{\rm dB} domain is given by
%\begin{align}
%	\mathcal{A}_{k} &\approx \text{BW}_{\rm 3dB}\times	\text{BD}_{\rm 3dB}\nn\\
%	& = 2 \times \frac{ 0.886\lambda}{Nd} \times \frac{2(r^{\rm B})^2r_{\rm DF}}{r^2_{\rm DF}-(r^{\rm B})^2}
%\end{align}
%}
%\end{proposition}

%\begin{align}
%	{\rm Ratio}	 &=  \frac{ \Omega_{\rm rec} }{ \Omega_{\rm ent} }\approx  \frac{ \text{BW}_{\rm 3dB} \times R_{\rm Ray}}{2  \times R_{\rm Ray}} \nn\\
%	& = \frac{ 0.886\lambda}{Nd} = \frac{1.772}{N}
%\end{align}
%Under the considered system setting, it is observed that the {\rm 3 dB} region is affected by the spatial angle and distance. Specifically, it can be verified that the {\rm 3 dB} region is monotonically increasing with the distance, while first increasing then decreasing at 
\subsubsection{Interference Set Construction} 
Based on the characterized interference domain, we can construct the interference set for each user accordingly.
%In the following, we elucidate the procedures of the proposed low-complexity scheme in detail.
To proceed, we provide the following definition.
 \begin{definition}[Interference set]
	\emph{For a legitimate user $k \in \mathcal{K}$, its interference set, defined as $\Xi_{k}$, consists of the legitimate users that are located within its $\varphi$-{\rm dB} interference domain.}
\end{definition}

In the following, we present two useful propositions for facilitating the interference set construction. The first one is conceived from Proposition \ref{Pro:BD}, where we observe that the beam depth tends to be infinity when the distance exceeds the boundary $r_{\rm BD}$. However, since we consider the near-field secure transmission, the beam-depth interval should also fall within the near-field region. To this end, we below establish a \emph{distance threshold}, denoted by $r_{\rm DT}$, whose right boundary of beam-depth interval exactly lies in the Rayleigh distance.
%
%In particular, we would like to point out that even when the distance is in proximity to, the right xx of the beam depth will be greater than the Rayleigh distance. 
%We below derive a \emph{threshold}, , outside of which the right boundary of beam depth exactly equals the Rayleigh distance.

 \begin{proposition}\label{Le:BR}
	\emph{For a legitimate user $k\in \mathcal{K}$, there exists a distance threshold $r_{\rm DT}$ in the distance domain, beyond which the right boundary of the beam-depth interval is greater than the Rayleigh distance, which is given by
		\begin{equation}\label{Eq:BR}
			r_{\rm DT} = \frac{N^2\lambda(1-\theta_{{\rm B},k}^2)}{8\beta_{\varphi}^2+2(1-\theta_{{\rm B},k}^2)}.
	\end{equation}
Specifically, the distance threshold has a maximum value with respect to $\theta_{{\rm B},k}$, which is attained when $\theta_{{\rm B},k}$ is equal to zero, given by
	\begin{align}\label{Eq:maxBR}
	r^{\max}_{\rm DT} = \frac{N^2\lambda}{8\beta_{\varphi}^2+2}.
	\end{align} }
\end{proposition}
\begin{proof}
Please refer to Appendix \ref{App4}.
\end{proof}
Next, we present another useful proposition below.
\begin{proposition}[Interference reciprocity]\label{Pro:two_set}
	\emph{For any two legitimate users $i,j \in \mathcal{K}$, if $j$ is in user $i$'s interference set, then user $i$ shall also be included in the interfering set of user $j$, i.e., 
		\begin{equation}
			j \in \Xi_{i}	\Longleftrightarrow i \in \Xi_{j},~~ \forall i,j \in \mathcal{K}.
	\end{equation}}
\end{proposition}
\begin{proof}
	Please refer to Appendix \ref{App5}.
\end{proof}

\begin{remark}[How do Propositions \ref{Le:BR} and \ref{Pro:two_set} facilitate the interference set construction?]\label{Re:xx}
	\emph{Proposition \ref{Le:BR} derives a distance threshold, beyond which the distinct beam-depth intervals tend to be overlapped heavily. Furthermore, by means of the results of Proposition \ref{Pro:two_set}, we can conclude that users whose distance is greater than the threshold will  definitely 
	interfere with each other, and therefore should be
	grouped into one interference set. The above conclusions have two benefits in promoting the interference set construction. First, the complexity of the interference set construction can be greatly reduced with a large portion of users being located beyond the distance threshold. While at the same time, the \emph{interference reciprocity}
	 makes it possible to find users contained in overlapping interference sets sequentially, thus avoiding repeated checks. More importantly, these two conclusions are naturally complementary, making the low-complexity approach concise
	 and intuitive.}
\end{remark}

Subsequently, since the interference set construction is largely affected by the user distribution, for ease of exposition, we divide it into two cases, as described below. 

\underline{\textbf{Linear distribution:}}
We commence with the challenging case, i.e., {all users lie in the same spatial angle}, to demonstrate the construction process of the interference set. Based on Proposition \ref{Le:BR}, taking $r_{\rm DT}$ as the boundary, the whole user set can be divided into two separate subsets, denoted as $\mathcal{K}_{\rm small}$ and $\mathcal{K}_{\rm large}$, which are given by
\begin{align}
	\mathcal{K}_{\rm small} &= \left\lbrace {i}|r_{i} < r_{\rm DT},  i \in \mathcal{K}\right\rbrace,\nn\\
	\mathcal{K}_{\rm large} &=\left\lbrace {i}|r_{i} \ge r_{\rm DT}, i \in \mathcal{K}\right\rbrace.
\end{align}
Then, according to Remark \ref{Re:xx}, all users in $\mathcal{K}_{\rm large}$ are interfering with each other, thereby being grouped into one interference set, i.e.,
 	\begin{equation}\label{Eq:allone}
	\Xi_{{\rm I}} \leftarrow \mathcal{K}_{\rm large}.
 \vspace{-3pt}
\end{equation}
 While for the remaining legitimate users whose distance is smaller than $r_{\rm DT}$, say user  $k \in \mathcal{K}_{\rm small}$, its interference set can be mathematically given as 
\begin{align}\label{Eq:part}
	\Xi_{k} = \left\lbrace {i}|r_{i}  \in  \text{BD}_{{\rm \varphi dB},k} , r_i \ge r_{k} \in \mathcal{K}_{\rm small}\right\rbrace.
\end{align}

\underline{\textbf{Uniform distribution:}} While for the general scenario where users are uniformly distributed within the near-field region, the interference set construction follows a similar idea as that of the linear distribution. More specifically, for a legitimate user $k \in \mathcal{K}_{\rm small}$, an additional step is to first find out which users are in its beam width, followed by the distance-domain identification as in \eqref{Eq:allone} and \eqref{Eq:part}. It is important to highlight that the proposed beam determination is independent of the user distribution once the interference sets are constructed. Accordingly, to provide a clearer illustration in determining the beamformer, we primarily focus on the case of linear distribution in the following. Notably, the beamformer determination procedures can be directly applied to the uniform distribution case.

\subsubsection{Beamformer Determination}
The two key factors that affect the system secrecy performance, namely inter-user interference and information leakage, can be measured by the cardinality of the interference sets and the minimum required power for security in Lemma \ref{Pro:sec_condi}, respectively. Ideally, for data transmission of a typical legitimate user, if the secure condition is met and its interference set is empty, it indicates that this user is secure and interference-limited. Under such circumstances, the MRT-based beamforming can be performed directly, i.e., making the beamformer aligned with its channel. Otherwise, the beamformer needs to be optimized using the conventional approach in Section \ref{Subsec:con}. To this end, the procedures of the beamformer determination can be divided into the following three cases:

%This low-complexity scheme is encapsulated in a two-layer hierarchical structure, in which xxx.
\textbf{Case 1 (No overlapping):} In this case, the interference sets of all legitimate users and $\Xi_{\rm I}$ are empty, indicating limited interference among users. As such, the beamformer determination of each legitimate user depends on the secure condition solely. In particular, the beamformers of legitimate users who require power (i.e., the corresponding $P_{{\rm E,min}}>0$) for security provisioning need to be optimized, and those of the remaining users shall employ MRT-based beamforming.
\begin{figure}[t]
	\centering
	\includegraphics[width=0.37\textwidth]{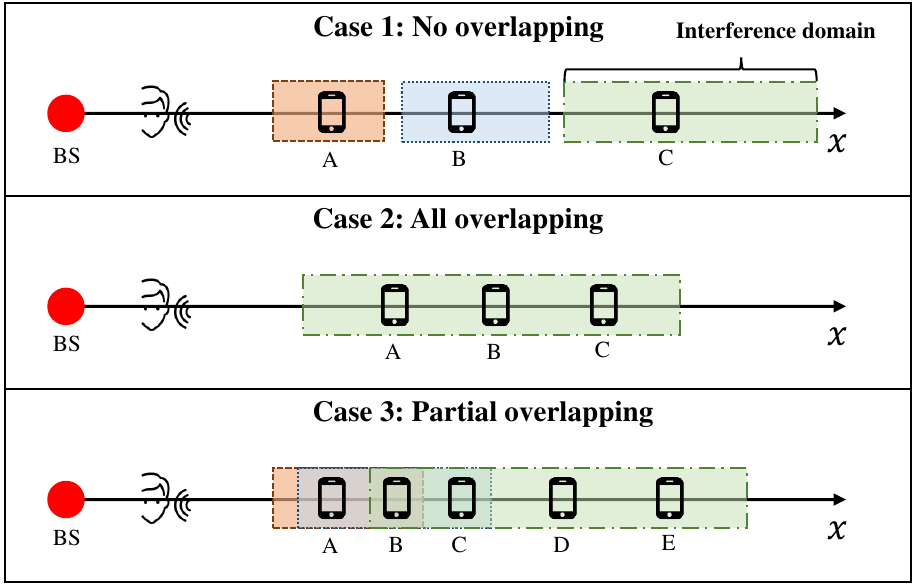}
	\caption{Three cases of beamformer determination.} \label{Fig:threecases}
 \vspace{-12pt}
\end{figure}

\textbf{Case 2 (Full overlapping):} For this case, all legitimate users are contained in an interference set, resulting in strong inter-user interference. Accordingly, we select one legitimate user with the minimum required power for security provisioning to conduct beamformer optimization. The remaining users adopt MRT-based beamforming directly. Particularly, it is worth noting that in practice, Cases 1 and 2 rarely occur, while the following Case 3 is relatively common.

\textbf{Case 3 (Partial overlapping):} When there are multiple non-empty interference sets, the beamformer determination becomes much more involved. The key challenge in the case of partial overlapping is identifying where the overlapping is and how many users are included in the overlapping area. To conquer this issue, we propose a \emph{three-step identification framework} for determining the beamformer optimization. Specifically, Fig. \ref{Fig:threecases} shows an illustrative example of multiple non-empty interference sets, i.e., $\Xi_{A}=\{B\}$, $\Xi_{B}=\{C\}$, and $\Xi_{C}=\{D, E\}$, to demonstrate the detailed procedures. 
    \begin{itemize}
        \item[1.] \emph{Identify the interference set of user $A$:} If it is empty, beamformer optimization is determined by the secure condition. Otherwise, identify the cardinality of the interference set of user $A$ and find the last user in the interference set $\Xi_{A}$, i.e., $B$. 
        \item[2.] \emph{Identify the interference set of user $B$:} Since the interference set of user $B$ is not empty, user $B$'s beamformer should be aligned. Meanwhile, user $A$'s beamformer optimization is determined by the secure condition. 
        \item[3.] \emph{Identify the interference set of user $C$:} The interference set of user $C$ is also not empty. Thus, user $C$'s beamformer should be aligned with its channel. Last, identify the interference set of the last user in $\Xi_{C}$, i.e., $E$, which is empty. Then, one legitimate user included in the interference set of $C$, i.e., $D$ or $E$,  is selected according to the secure condition for beamformer optimization.
    \end{itemize}

\begin{remark}[How to design the beamformer of AN?]\label{Re:AN_Align}\emph{
%	AN is capable of degrading the signal power leaked at the eavesdropper, while at the same time causing interference to the legitimate users. 
		In particular,
		it has been shown in Example \ref{Ex:AN} that by leveraging the beam focusing effect, AN can be deliberately adjusted to be focused on the eavesdropper to effectively suppress leaked signal power. At the same time, AN causes limited interference to legitimate users due to the worse channel conditions. Therefore, we can infer that for the design of near-field AN beamformer, it should be properly pointed towards the eavesdropper to achieve near-optimal secrecy performance. The corresponding performance will be evaluated via simulations in Section \ref{Sec:SR}. }
\end{remark}
\begin{figure}[t]
	\centering	\includegraphics[width=0.38\textwidth]{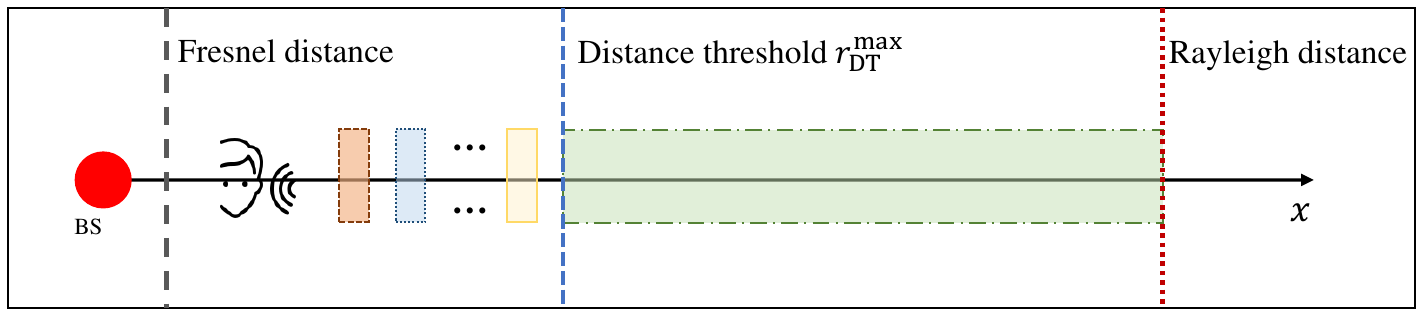}
	\caption{Extreme conditions.} \label{Fig:complex}
 \vspace{-12pt}
\end{figure}
\vspace{-12pt}
\subsection{Complexity Analysis}\label{Sec:Complexity}
Next, we compare the computational complexity of the proposed approach with that of the conventional method in Section \ref{Subsec:con}.
%Based on the procedures xx, the worst-case complexity of the proposed low-complexity xx scheme can be characterized as follows.
%%
%,  It can be concluded that the complexity of the proposed low-complexity scheme is largely determined by the user distribution, i.e., interfering domain. Notice that when all users are separable enough, i.e., very small correlation, the system performance is asymptotic optimal.
%To characterize the complexity , we provide xxx below. Note that the 
To be more specific, consider the worst-case computational complexity of our approach, assuming that the following extreme conditions are met.
 First, consider the scenario where users are located at the same spatial angle, and all users are distinguished by the defined boundary in Proposition \ref{Le:BR}, as shown in Fig. \ref{Fig:complex}. Furthermore, it is assumed that users located on the left side of the boundary are sufficiently separable, i.e., no overlapping interference sets. Finally, we adopt the farthest boundary in \eqref{Eq:maxBR} such that the \emph{separable region} is large enough. In this way, the computational complexity is determined by the geometric distribution of legitimate users, which is derived as follows.
 \begin{proposition}\label{Pro:far}
 \emph{The distribution of the number of legitimate users that require beamformer optimization can be characterized as
\begin{align}
    &P(K=m)=	\\
&\begin{cases}
		{K \choose m-1}( P_{\rm sr}) ^{m-1} \nn
 		( 1-P_{\rm sr})^{K-m+1}, & m=1,\cdots,K-1,\\
		{K \choose K-1}( P_{\rm sr}) ^{K-1} \nn
 		( 1-P_{\rm sr})+(P_{\rm sr})^{K},  &m=K.
	\end{cases}
\end{align}
Accordingly, the expectation of the number of legitimate users that the beamformer needs to be optimized is given by
 \begin{equation}
K_{\rm req} = 1+KP_{\rm sr}-(P_{\rm sr})^{K},
 \end{equation}
 where $P_{\rm sr}=\frac{r^{\rm max}_{\rm DT}-r_{\rm Fre}}{r_{\rm Ray}-r_{\rm Fre}}$ denotes the probability that the user is positioned in the separable region.}
 \end{proposition}
As such, the worst-case complexity of the proposed approach can be characterized as $\mathcal{O}(	K_{\rm req} N^{6.5})$ \cite{2010interior}. While for the conventional approach, it entails a computational complexity of $\mathcal{O}((K+1)N^{6.5})$. For instance, when there are $K=20$ uniformly distributed legitimate users along one specific angle, the average number for beamformer optimization is $K_{\rm req}=2.2691$. This indicates that the computational complexity of the proposed approach is about at most $11.3\%$ of that of the conventional approach, rendering its great advantages in practical implementation.

\vspace{-6pt}
\section{Simulation Results}\label{Sec:SR}
In this section, we present numerical results to demonstrate the performance gains of the near-field PLS over its far-field counterpart, as well as the effectiveness of the
proposed low-complexity beamforming scheme. The system parameters are set as follows. The BS equipped with $N=256$ antennas operates at a frequency of $f=100$ GHz to serve $K=3$ legitimate users. We consider the most representative and challenging scenario where the eavesdropper is positioned at $(0,0.05r_{\rm Ray})$, and three legitimate users are located at $(0,0.06r_{\rm Ray})$, $(0,0.1r_{\rm Ray})$, and $(0,0.3r_{\rm Ray})$. Unless specified otherwise, the system parameters are described below. The reference path loss is $\beta=(\lambda/4\pi)=-72$ dB, the maximum transmit power is  $P=30$ dBm and the noise power is set as $\sigma^2_{{\rm B},k}=\sigma^2_{\rm{E}} = -80$ dBm, $\forall k \in \mathcal{K}$, and $\varphi=3$. Moreover, the number of NLoS paths from the BS to legitimate user/eavesdropper is set to $L_{{\rm B},k}=2, \forall k\in \mathcal{K}$ and $L_{\rm E}=2$. The channel gains $h_{{\rm B},k,\ell}$ and $h_{{\rm E},\ell}$ are modeled as $h_{{\rm B},k,\ell}\sim \mathcal{CN}(0,\eta^2_{{\rm B},k,\ell})$ and $h_{{\rm E},\ell}\sim \mathcal{CN}(0,\eta^2_{{\rm E},\ell})$, respectively, where $\eta_{{\rm B},k,\ell}=\kappa h_{{\rm B},k}$, $\eta_{{\rm E},\ell}=\kappa h_{{\rm E}}$, and $\kappa = -15$ dB \cite{wang2023beamfocusing}.
Specifically, we consider the following benchmark schemes for performance comparison.
\begin{itemize}
\item  \emph{Conventional scheme:} It performs the conventional approach to obtain the joint beamforming design with AN, thus serving as an upper bound for performance comparison.
\item \emph{Baseline scheme 1:} The joint transmit beamforming with AN is configured to point towards their respective users, while only power allocations are optimized.
\item  \emph{Baseline scheme 2:} for which the transmit beamformer is optimized while the AN  is directly aligned with the eavesdropper's channel according to Remark \ref{Re:AN_Align}.
\item \emph{Baseline scheme 3:} for which the transmit beamformers of all legitimate users are aligned with the associated channels while optimizing the AN only.
\end{itemize}
\vspace{-9pt}
\subsection{Effect of Maximum Transmit Power}
In Fig. \ref{Fig:rate_power}, we study the effect of maximum transmit power on the achievable secrecy rate by different schemes. First, we observe that all involved schemes can guarantee secure transmission in the considered scenario, demonstrating the superiority of beam focusing in PLS improvement compared to the far-field counterpart. Besides, 
the secrecy rates by all schemes monotonically increase with the transmit power. Second, it is observed that baseline scheme 1 and the proposed approach achieve a very close secrecy rate to that of the conventional method, which is consistent with Remark \ref{Re:AN_Align}. It is worth noting that the conventional scheme is with prohibitively high computational complexity, whereas the proposed method achieves significantly lower computational complexity, rending its substantially higher efficiency.
Next, the proposed scheme attains substantial performance gain over other benchmark methods, e.g., around twice the secrecy rate by the baseline scheme 3. Last, it is observed that the secrecy rate of baseline scheme 3 tends to be saturated in the high transmit power regime. This is because the adopted MRT beamforming cannot exploit the spatial-domain DoFs, hence being unable to mitigate the multi-user interference effectively.

\vspace{-6pt}
\subsection{Effect of Number of Legitimate Users}
In Fig. \ref{Fig:rate_num_bob}, we evaluate the secrecy rate of the proposed scheme under various numbers of legitimate users. Specifically, $K$ legitimate users are sequentially selected from the distance interval $[0.25r_{\rm Ray},0.35r_{\rm Ray}]$ in descending order, while all of them are positioned at the same angle $\theta=0$. First, we see that the secrecy rates of all schemes increase with the number of legitimate users. This is because the newly added legitimate user with more favorable channel conditions is selected for data transmission, resulting in a higher secrecy rate. Second, it is observed that under different numbers of legitimate users, the proposed scheme is significantly superior to benchmark schemes 1 and 3, with a very marginal performance loss compared to the conventional method.
\vspace{-9pt}
\subsection{Effect of Spatial Angle}
Next, we show the effect of the spatial angle on the secrecy rate by different schemes in Fig. \ref{Fig:rate_angle}. For clarity, we consider the scenario with one eavesdropper located at $(0,0.05r_{\rm Ray})$ and two legitimate users. One of them is situated at $(0,0.3r_{\rm Ray})$, while the other one is placed at a distance of $0.3r_{\rm Ray}$, with the spatial angle varying between $[-0.2,0.2]$. First, it is observed that the secrecy rates by all schemes experience severe performance degradation when the two legitimate users overlap with each other. This is because in this case, the inter-user interference is the most intense, and the information leakage is the strongest, thus resulting in the lowest secrecy rate. Moreover, an intriguing observation is that the secrecy performance of baseline schemes 1 and 3 experiences fluctuations across the entire spatial angle interval, In contrast, the proposed scheme only exhibits fluctuations within the large-angle regime, e.g., $|\theta| \ge 0.1$, but with a moderate amplitude. This is because baseline schemes 1 and 3 adopt the MRT-based beamforming without utilizing the abundant spatial DoFs, rendering them ineffective in mitigating the inter-user interference. While for the proposed scheme, the MRT beamforming is employed only when the angular separation between two legitimate users exceeds the beam width and the user with a changing angle is secure.
\begin{figure}[t!]
	\centering
	\includegraphics[width=0.385\textwidth]{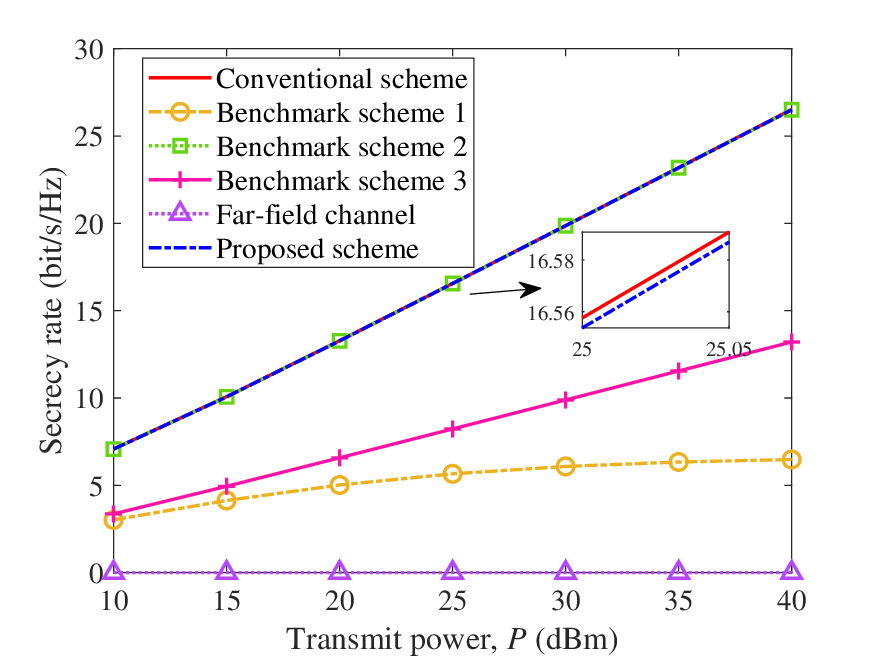}
	\caption{Effect of BS transmit power with $K=3$, $P=30$ dBm.} \label{Fig:rate_power}
    \vspace{-9pt}
\end{figure}
\begin{figure}[t!]
	\centering
	\includegraphics[width=0.385\textwidth]{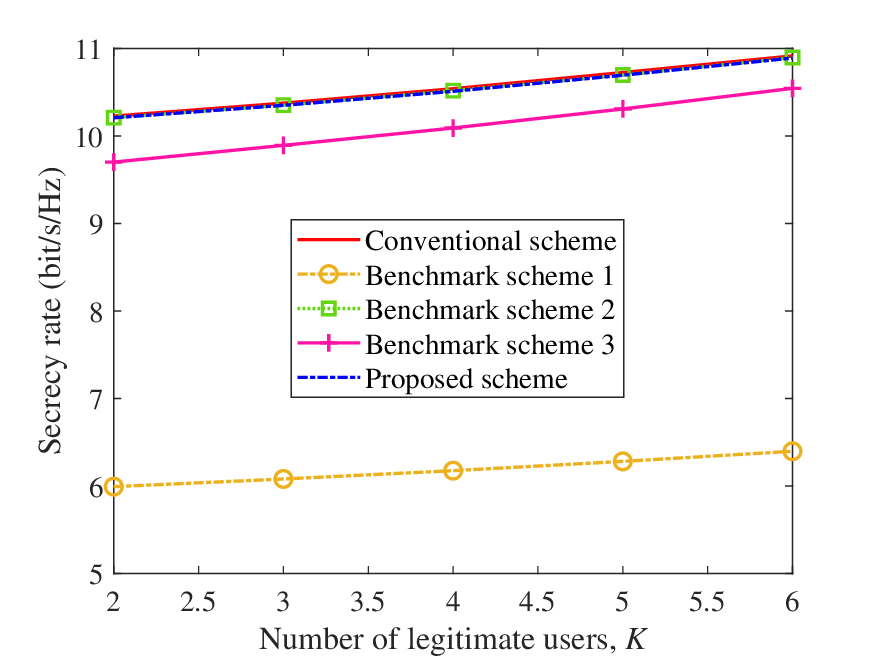}
	\caption{Effect of number of legitimate users with $P=30$ dBm.} \label{Fig:rate_num_bob}
  \vspace{-9pt}
\end{figure}
\begin{figure}[t!]
	\centering	\includegraphics[width=0.385\textwidth]{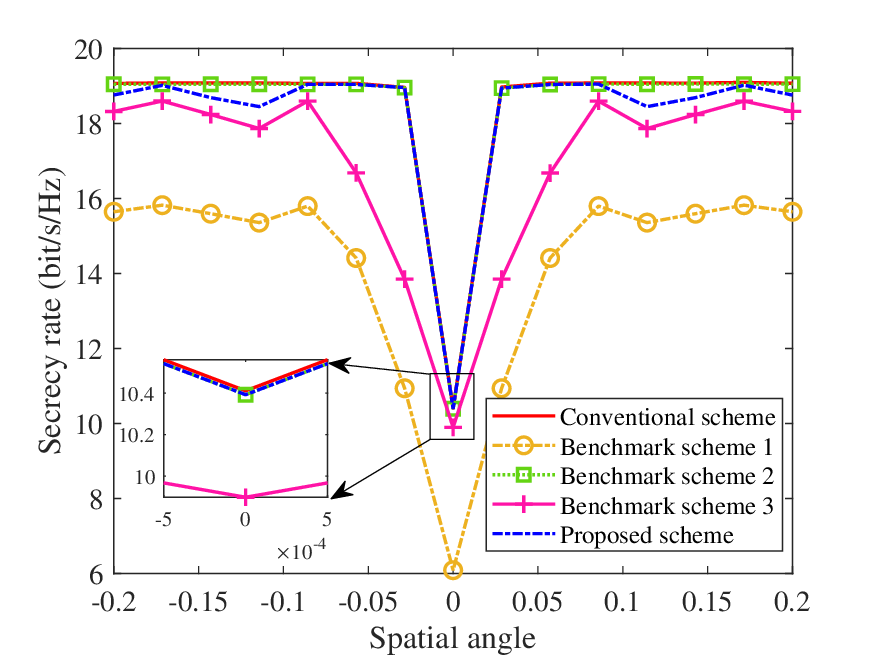}
	\caption{Effect of user spatial angle with $K=2$, $P=30$ dBm.} \label{Fig:rate_angle}
   \vspace{-9pt}
\end{figure}
\begin{figure}[t!]
	\centering
\includegraphics[width=0.385\textwidth]{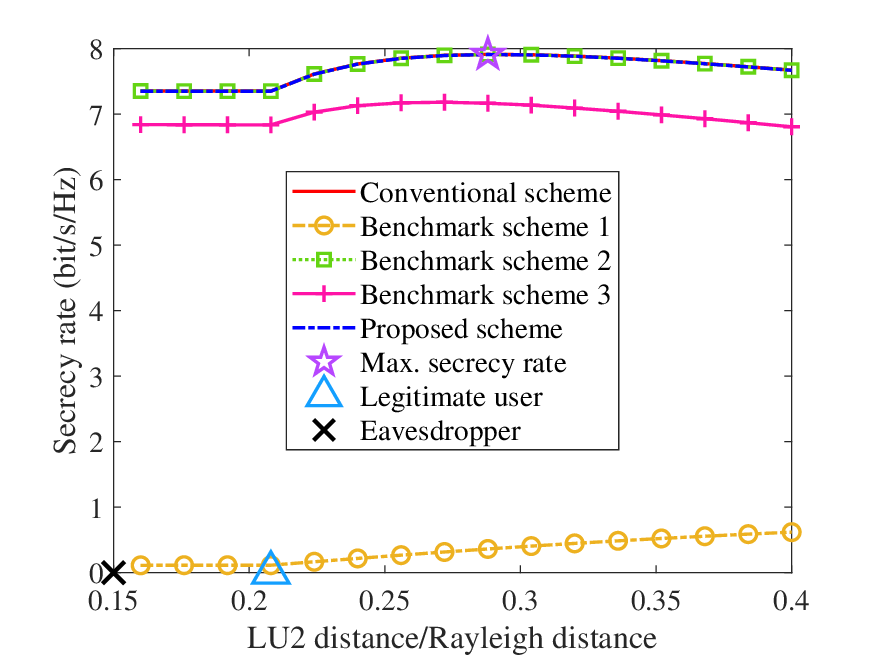}
	\caption{Effect of user distance with $K=2$, $P=30$ dBm.} \label{Fig:rate_dist}
   \vspace{-9pt}
\end{figure}
 \subsection{Effect of Distance}
In Fig. \ref{Fig:rate_dist}, we illustrate the effect of legitimate user distance on the secrecy rate. Particularly, we consider a scenario consisting of one fixed eavesdropper (marked with $\boldsymbol{\times}$), one fixed legitimate user, LU1 (marked with triangle), and one position-changing legitimate user LU2. First, we can observe that the secrecy rates of all schemes remain unchanged when the LU2 distance is smaller than that of LU1, indicating that all power allocated to LU1. However, when the LU2 distance is greater than the LU1 distance, the secrecy rate of benchmark scheme 1 shows a monotonous increase. While the secrecy performance of other schemes tends to increase first and then decrease, reaching a maximum marked with a purple pentagram. This is attributed to the non-linearity coupled effects of the channel gain and the correlation in terms of the LU2 distance, as stated in Proposition \ref{apppro1}. 
\vspace{-6pt}
 \subsection{Complexity Comparison}
Last, we compare the runtime of the proposed method with that of the conventional scheme under different numbers of the legitimate users. Specifically, the simulation is conducted using MATLAB
on a computer equipped with Intel Core i$7$-$1300$ and $3.40$ GHz
processor. It is clearly observed that the runtime of the proposed scheme is much shorter than the conventional method under different values of $K$. Additionally, the runtime reduction achieved by the proposed scheme becomes more significant as $K$ increases. This is because the number of users requiring beamformer optimization tends to saturate in the large-$K$ regime, which aligns with the result in Section \ref{Sec:Complexity}.
% \subsection{Effect of the Distance}
% Last, 
% Can show the effect of relative distance in this part?
%\emph{Linear Distribution:}  For example, from the perspective of rate maximization, the one with best channel condition is more likely to xxx. While on other hand, if this user locates in the neighborhood of the eavesdropper, allocating more power will in turn result in a higher information leakage, thus worsening the ultimate secrecy performance. To xxx, we propose an xxx, consisting of xxx
\begin{table}[t]
\centering 
\caption{Computational Complexity Comparison}\label{complexity}
\begin{tabular}{c|ccc}
\hline
\hline
\multirow{2}{*}{} & \multicolumn{3}{c}{Runtime (s)} \\ \cline{2-4} & $K=2$  & $K=3$ & $K=4$ \\ \hline
Conventional scheme& $535$ & $2256$ &$4177$\\ \hline
 Proposed scheme& $\mathbf{268}$ & $\mathbf{1020}$ &$\mathbf{1524}$ \\ \hline
\end{tabular}
\vspace{-12pt}
\end{table}

\section{Conclusions}\label{Se:Con}
In this paper, we studied near-field secure communication in a multi-user MISO system, wherein multiple legitimate users and one eavesdropper are assumed to be located in the near-field region. 
% A secrecy rate maximization problem was formulated by optimizing the joint transmit beamforming with AN under the sum-power constraint. 
To shed useful insights into the design of beam focusing-based near-field PLS, we first considered a special case with one legitimate user and one eavesdropper. In particular, it was unveiled that 1) AN is of great importance to the system security provisioning, making insecure systems secure again; 2) AN can bring pronounced performance gains using a moderate portion of power. Moreover, for the general case, we first introduced a suboptimal solution by invoking the SCA and SDR techniques. Subsequently, an implementation-friendly low-complexity approach tailored for beamformer design was proposed based on the introduced interference domain and the three-step identification framework. Numerical results are presented to illustrate the performance gains achieved by near-field PLS compared to the far-field counterpart, as well as the effectiveness of the proposed low-complexity approach. 
In the future, it is interesting to devise more computationally-efficient near-field beamformers to further improve the near-field PLS.

\begin{appendices}
	\section{Proof of Lemma \ref{Pro:sec_condi}}\label{App1}
 To verify Lemma \ref{Pro:sec_condi}, we first state a fact that all transmit power should be utilized for maximizing the secrecy rate \cite{10054084}.
	% To verify Lemma \ref{Pro:sec_condi}, we first show a useful lemma below.
	% \begin{lemma}\label{Le:fullpow}
	% 	\emph{The power allocations to the legitimate user and eavesdropper should satisfy the sum-power constraint with equality, which is given by
	% 	\begin{equation}
	% 			P_{\rm B}+P_{\rm E}=P.
	% 	\end{equation}}
	% \end{lemma}
	% \begin{proof}
	% 	The proof is similar to that of \cite{10054084}, and thus is omitted for brevity.
	% \end{proof}
		Then, the secrecy rate $R^{\rm Sec}$ in \eqref{Eq1:secrecyrate_AN} can be expressed as 
	\begin{equation}\label{Eq:newsecrate}
		R^{\rm Sec} = \log_2\left(\frac{A+\tilde{B}}{A+\tilde{C}} \right) 
	\end{equation}
	where
	\begin{equation}
		\begin{cases}
			\tilde{B}=-P_{\rm E}^2g_{\rm B}g_{\rm E}+P_{\rm E}(Pg_{\rm B}g_{\rm E}-g_{\rm B}\sigma^2)+Pg_{\rm B}\sigma^2, \\ 
			\tilde{C}=-P_{\rm E}^2g_{\rm B}g_{\rm E}\rho^4+P_{\rm E}(Pg_{\rm B}g_{\rm E}\rho^4-g_{\rm E}\rho^2\sigma^2)+Pg_{\rm E}\rho^2\sigma^2.
		\end{cases}
	\end{equation} 
	For ease of analysis, the secrecy rate in \eqref{Eq:newsecrate} can be further rewritten as:
	\begin{equation}\label{Eq:NE}
		R^{\rm Sec} =\log_{2}\left(1+ \frac{\tilde{B}-\tilde{C}}{A+\tilde{C}}\right),
	\end{equation} 
	Next, we only need to 
find the condition, under which $\tilde{B}-\tilde{C}= 0$ holds.  
	%	\begin{proposition}\label{Pr:secure}
		%		\emph{The system is always insecure, i.e., $R\le 0$, if we have
			%			\begin{equation}
				%				\rho \ge \frac{Pg_{\rm B}(1-\rho^2)}{\sigma^2}+\frac{g_{\rm B}}{g_{\rm E}}.
				%		\end{equation}}
		%	\end{proposition}
	%	\begin{proof}
	Therefore, we define the following function
		\begin{align}\label{Eq:delta}
			\xi = \tilde{B}-\tilde{C}
			%			&=P_{\rm B}P_{\rm E}g_{\rm B}g_{\rm E}+P_{\rm B}g_{\rm B}\sigma^2-(P_{\rm B}P_{\rm E}g_{\rm B}g_{\rm E}\rho^2+P_{\rm B}g_{\rm E}\rho\sigma^2)\nn\\
%			=P_{\rm B}P_{\rm E}g_{\rm B}g_{\rm E}(1-\rho^4)+P_{\rm B}(g_{\rm B}-g_{\rm E}\rho^2)\sigma^2.\label{Eq1}
&=-P_{\rm E}^2g_{\rm B}g_{\rm E}(1-\rho^4)+PP_{\rm E}g_{\rm B}g_{\rm E}(1-\rho^4)\nn\\
&-P_{\rm E}(g_{\rm B}-g_{\rm E}\rho^2)\sigma^2
+P(g_{\rm B}-g_{\rm E}\rho^2)\sigma^2.
		\end{align}
%	
%		
%		\begin{align}\label{Eq:delta}
%			\xi 
%%			P_{\rm B}P_{\rm E}g_{\rm B}g_{\rm E}(1-\rho^2)+P_{\rm B}\sigma^2(g_{\rm B}-g_{\rm E}\rho)\nn\\
%			%			&=(P-P_{\rm E})P_{\rm E}g_{\rm B}g_{\rm E}(1-\rho^2)+(P-P_{\rm E})\sigma^2(g_{\rm B}-g_{\rm E}\rho)\nn\\
%			=&-P_{\rm E}^2g_{\rm B}g_{\rm E}(1-\rho^4)+PP_{\rm E}g_{\rm B}g_{\rm E}(1-\rho^4)\nn\\
%			&-P_{\rm E}(g_{\rm B}-g_{\rm E}\rho^2)\sigma^2
%			+P(g_{\rm B}-g_{\rm E}\rho^2)\sigma^2.
%		\end{align}
		It is clearly observed that function $\xi$ in \eqref{Eq:delta} is a quadratic function with respect to $P_{\rm E}$, whose discriminant can be expressed as
		%		\begin{align}
			%			\Delta = \sqrt{Pg_{\rm B}g_{\rm E}(1-\rho^2)+\sigma^2(g_{\rm B}-g_{\rm E}\rho)}.
			%		\end{align}
		\begin{align}\label{Eq:discri}
			\Delta = \left[ Pg_{\rm B}g_{\rm E}(1-\rho^4)+( g_{\rm B}-g_{\rm E}\rho^2)\sigma^2 \right]^2.
		\end{align}
		Due to the non-negativity of $\Delta$, that is, $\Delta\ge 0$, the zeros of function $\xi$ always exist.
		Then, in order to obtain the zeros analytically, defined as $P^{(1)}_{\rm E}$ and $P^{(2)}_{\rm E}$, the following two cases need to be considered:
		\begin{itemize}
			\item[1)] \emph{Same Zeros:} In this case, by letting $\Delta=0$, we have 
			\begin{equation}
				g_{\rm B} - g_{\rm E}\rho^2 = -\frac{Pg_{\rm B}g_{\rm E}(1-\rho^4)}{\sigma^2}. \vspace{-3pt}
			\end{equation}
			Accordingly, $\xi$ has two identical zeros, which are given by
			\begin{equation}
				P_{\rm E}^{(1)}=P_{\rm E}^{(2)}=P.
            \vspace{-3pt}
			\end{equation}
			\item[2)]  \emph{Different Zeros:}  
			For this case, we have $\Delta >0$, then the two distinct zeros of function $\xi$ are given by
			\begin{equation}
				\begin{cases}
					P_{\rm E}^{(1)}=\frac{(g_{\rm E}\rho^2-g_{\rm B})\sigma^2}{g_{\rm B}g_{\rm E}(1-\rho^4)},\\
					P_{\rm E}^{(2)}=P.
				\end{cases}
			\end{equation}
			Since $P_{\rm E}^{(1)}$ is always smaller than $P_{\rm E}^{(2)}$, we only consider the case $P_{\rm E}^{(1)}<P_{\rm E}^{(2)}$, i.e., $	g_{\rm B} - g_{\rm E}\rho^2 > -\frac{Pg_{\rm B}g_{\rm E}(1-\rho^4)}{\sigma^2}$.  
			%The detailed procedures can be divided into two cases, which are provided below.
			%				\item[i)] If $	g_{\rm B} - g_{\rm E}\rho < -\frac{Pg_{\rm B}g_{\rm E}(1-\rho^2)}{\sigma^2}$ holds, we have $P_{\rm E}_{1}>P_{\rm E}_{2}$. According to \eqref{Eq:delta}, it is shown that $\xi \le 0$ over the whole feasible region, i.e., $	g_{\rm B} - g_{\rm E}\rho > -\frac{Pg_{\rm B}g_{\rm E}(1-\rho^2)}{\sigma^2}$ $0 \le P_{\rm E}\le P$, hence resulting in an insecure system.
			%					\begin{equation}
				%					g_{\rm B} - g_{\rm E}\rho < -\frac{Pg_{\rm B}g_{\rm E}(1-\rho^2)}{\sigma^2}.
				%				\end{equation}
			%While if  $	g_{\rm B} - g_{\rm E}\rho > -\frac{Pg_{\rm B}g_{\rm E}(1-\rho^2)}{\sigma^2}$ holds, we have $P_{\rm E}_{1}<P_{\rm E}_{2}$.
			%					\begin{equation}
				%					g_{\rm B} - g_{\rm E}\rho > -\frac{Pg_{\rm B}g_{\rm E}(1-\rho^2)}{\sigma^2}.
				%				\end{equation}
			Notice that in this case, $P_{\rm E}^{(1)}$ is either greater than zero or less than zero. Accordingly, when $P_{\rm E}^{(1)}<0$, we can obtain $g_{\rm B} - g_{\rm E}\rho^2 \ge 0$, hence indicating that the system security can be guaranteed at will. This fact is consistent with Proposition \ref{Pro:secure_reg}. Next, for the case $P_{\rm E}^{(1)}>0$, i.e., $g_{\rm B} - g_{\rm E}\rho^2 < 0$,
			we can infer that the value of $P_{\rm E}^{(1)}$ is between zero and $P_{\rm E}^{(2)}$, thus unveiling the minimum power required to achieve system security.
			
			%			 it can be easily verified that function $\xi$ posses the maximum value while has no root as shown in Fig. \ref{label}. As such, function $\xi$ is always smaller than zero for any feasible $P_{\rm E}$, i.e., $\xi < 0,~\forall P_{\rm E}.$ This thus means that the consider system is always insecure regardless of the power allocation schemes. 
		\end{itemize}	
		Combining the above, we can derive the secure condition, i.e., $		g_{\rm B}-g_{\rm E}\rho^2 \ge -\frac{Pg_{\rm B}g_{\rm E}(1-\rho^4)}{\sigma^2}$, which thus completes the proof.	
		%		It is clearly observed that $\Delta$ has two roots, which are given by
		%		\begin{equation}
			%			P_{\rm B}=
			%			\begin{cases}
				%				0  \\
				%				\frac{Pg_{\rm B}g_{\rm E}(1-\rho^2)+\sigma^2(g_{\rm B}-g_{\rm E}\rho)}{g_{\rm B}g_{\rm E}(1-\rho^2)}
				%			\end{cases}.
			%		\end{equation}
		%		Then, there are two possible cases as elaborated.
		%		\begin{itemize}
			%			\item[1)] there are one root lies in the left side of the axis. This means that the following condition holds
			%			\begin{equation}
				%				g_{\rm E}\rho-g_{\rm B} > \frac{Pg_{\rm B}g_{\rm E}(1-\rho^2)}{\sigma^2}.
				%			\end{equation}
			%			In this case, $\Delta$ has only one feasible zero point, i.e., $P_{\rm B}=0$.
			%			\item[2)] Otherwise, when the following condition holds,
			%			\begin{equation}
				%				g_{\rm E}\rho-g_{\rm B} \le \frac{Pg_{\rm B}g_{\rm E}(1-\rho^2)}{\sigma^2}.
				%			\end{equation}
			%			$\Delta$ first increases and then decreases, for which $\Delta=0$ when $	P_{\rm B}=\frac{Pg_{\rm B}g_{\rm E}(1-\rho^2)+\sigma^2(g_{\rm B}-g_{\rm E}\rho)}{g_{\rm B}g_{\rm E}(1-\rho^2)}$. 
			%		\end{itemize}
		%		
		%		
		%		the condition for $\Delta \le 0$ can be directly obtained by letting the second root be smaller than zero, which is given by
		%		\begin{equation}
			%			f \ge \frac{Pg_{\rm B}(1-\rho^2)}{\sigma^2}+\frac{g_{\rm B}}{g_{\rm E}}.
			%		\end{equation} 

%%%%%%%%%%%%%%%%%%%
\vspace{-4pt}
    \section{Proof of Lemma \ref{Le:opt_pow}}\label{App2}
    First, we rewrite the secrecy rate as $R^{\rm Sec}=\log_{2}(\zeta)$, where $\zeta = \frac{A+\tilde{B}}{A+\tilde{C}}$. 
Then, obtaining the maximum secrecy rate is equivalent to finding the maximum value of function $\zeta$. To this end, we first obtain the first-order derivative of $\zeta$, which is given by $\frac{\partial \zeta}{\partial P_{\rm E}}$.
% \begin{align}
% 	\frac{\partial \zeta}{\partial P_{\rm E}}=\frac{(A+\tilde{B})^{'}(A+\tilde{C})-(A+\tilde{C})^{'}(A+\tilde{B})}{(A+\tilde{C})^{2}},
% \end{align}
Notice that the monotonicity of the first derivative function above relies only on its numerator, so we study its numerator instead, which is given by,
	\begin{align}\label{Eq:numer_fir_der}
		f(P_{\rm E})
%		(A+\tilde{B})^{'}(A+\tilde{C})-(A+\tilde{C})^{'}(A+\tilde{B})
%		&=\left[ P\rho^3   (g_{\rm B})^2   (g_{\rm E})^2 +\sigma^2  \rho^3   (g_{\rm B})^2 g_{\rm E}-\sigma^2  \rho^2   (g_{\rm B})^2  g_{\rm E} -P \rho (g_{\rm B})^2   (g_{\rm E})^2 +\sigma^2  \rho g_{\rm B} (g_{\rm E})^2 
%		-\sigma^2  g_{\rm B}  (g_{\rm E})^2 \right] (P_{\rm E})^2 \nn\\
%		&~~~+\left[2 P \sigma^2  \rho^2   (g_{\rm B})^2  g_{\rm E}+2 P \sigma^2  \rho^2 g_{\rm B}  (g_{\rm E})^2+2\rho^2g_{\rm B}g_{\rm E}\sigma^4 -2 P \sigma^2  \rho  (g_{\rm B})^2  g_{\rm E} -2 P \sigma^2  \rho g_{\rm B}  (g_{\rm E})^2-2g_{\rm B}g_{\rm E}\sigma^4 \right] P_{\rm E}\nn\\
%		&~~~-P^2  \rho^2   (g_{\rm B})^2  g_{\rm E}\sigma^2 +P^2  \rho g_{\rm B} (g_{\rm E})^2  \sigma^2  -P \rho  (g_{\rm B})^2  \sigma^4+P \rho (g_{\rm E})^2  \sigma^4 +\rho g_{\rm E}\sigma^6-g_{\rm B}\sigma^6\nn\\
		% =\left[ P  g_{\rm B}^2   g_{\rm E}^2\rho^2  (\rho^4-1)+ g_{\rm B}^2 g_{\rm E} (\rho^2-1)\rho^4  \sigma^2 +  g_{\rm B}g_{\rm E}^2(\rho^2-1)\sigma^2 \right] P_{\rm E}^2 \nn\\
		% &~~~+\left[2 P    g_{\rm B} g_{\rm E}[g_{\rm B}(\rho^2-1) \rho^2  \sigma^2+g_{\rm E}(\rho^2-1)]+2g_{\rm B}g_{\rm E}(\rho^4-1)\sigma^4\right] P_{\rm E}\nn\\
		% &~~~+P^2    g_{\rm B}  g_{\rm E} (g_{\rm E}-g_{\rm B}\rho^2) \rho^2\sigma^2+P(g_{\rm E}^2 -g_{\rm B}^2) \rho^2 \sigma^4+ (g_{\rm E}\rho^2 -g_{\rm B})\sigma^6\nn\\
		\triangleq\Upsilon_1 P_{\rm E}^2+\Upsilon_2P_{\rm E}+\Upsilon_3,
	\end{align}
%	\begin{align}
%	&f=(A+\tilde{B})^{'}(A+\tilde{C})-(A+\tilde{C})^{'}(A+\tilde{B})\nn\\
%	&=\left[ P\rho^3   (g_{\rm B})^2   (g_{\rm E})^2 +\sigma^2  \rho^3   (g_{\rm B})^2 g_{\rm E}-\sigma^2  \rho^2   (g_{\rm B})^2  g_{\rm E} \right.\\  &\left.-P \rho (g_{\rm B})^2   (g_{\rm E})^2 +\sigma^2  \rho g_{\rm B} (g_{\rm E})^2 
%	-\sigma^2  g_{\rm B}  (g_{\rm E})^2 \right] (P_{\rm E})^2 \nn\\
%	&+\left[2 P \sigma^2  \rho^2   (g_{\rm B})^2  g_{\rm E}+2 P \sigma^2  \rho^2 g_{\rm B}  (g_{\rm E})^2+2\rho^2g_{\rm B}g_{\rm E}\sigma^4 \right.\\&\left.-2 P \sigma^2  \rho  (g_{\rm B})^2  g_{\rm E} -2 P \sigma^2  \rho g_{\rm B}  (g_{\rm E})^2-2g_{\rm B}g_{\rm E}\sigma^4 \right] P_{\rm E}\nn\\
%	&-P^2  \rho^2   (g_{\rm B})^2  g_{\rm E}\sigma^2 +P^2  \rho g_{\rm B} (g_{\rm E})^2  \sigma^2  -P \rho  (g_{\rm B})^2  \sigma^4 \nn\\
%	&+P \rho (g_{\rm E})^2  \sigma^4 +\rho g_{\rm E}\sigma^6-g_{\rm B}\sigma^6\nn\\
%	&=\left[ P\rho   (g_{\rm B})^2   (g_{\rm E})^2 (\rho^2-1)+\sigma^2  \rho^2   (g_{\rm B})^2 g_{\rm E}(\rho-1)\right.\\&\left.+\sigma^2  g_{\rm B}(g_{\rm E})^2(\rho-1) \right] (P_{\rm E})^2 \nn\\
%	&+\left[2 P \sigma^2  \rho   (g_{\rm B}) g_{\rm E}[(\rho-1)g_{\rm B}+(\rho-1)g_{\rm E}]\right.\\&\left.+2g_{\rm B}g_{\rm E}\sigma^4(\rho^2-1)\right] P_{\rm E}\nn\\
%	&+P^2  \rho   g_{\rm B}  g_{\rm E}\sigma^2 (g_{\rm E}-g_{\rm B}\rho)+P \rho \sigma^4((g_{\rm E})^2 -(g_{\rm B})^2) \nn\\
%	&+\sigma^6 (g_{\rm E}\rho -g_{\rm B})\nn\\
%	&=\Upsilon_1 (P_{\rm E})^2+\Upsilon_2P_{\rm E}+\Upsilon_3,
%\end{align}
First, it is shown that function $f(P_{\rm E})$ is a quadratic function with respect to $P_{\rm E}$.
Then, it can be easily verified that the correlation $\rho$ is smaller than one, i.e., $0<\rho\le1$, based on which
we can easily obtain that $\Upsilon_1 $ and $\Upsilon_2 $ are always smaller than zero for various system settings, i.e. $\Upsilon_1<0$ and $\Upsilon_2<0$. Accordingly, the property of function $f(P_{\rm E})$ is determined by the sign of $\Upsilon_3$, where both the first and second terms appear to be positive under the settings considered. While for the last term, its sign can be divided into two cases. First, when $\Upsilon_3>0$, i.e., $g_{\rm B}-g_{\rm E}\rho^2< \frac{P^2 g_{\rm B}g_{\rm E}(g_{\rm E}-g_{\rm B}\rho^2)\rho^2+P(g_{\rm E}^2-g_{\rm B}^2)\rho\sigma^2}{\sigma^4}$, then it is shown that the function $f(P_{\rm E})$ is first greater than zero and then smaller than zero, thereby resulting in the secrecy rate first increasing and then decreasing. Therefore, the maximum secrecy rate does exist and is obtained at the zero of function $f(P_{\rm E})$, which is given by 
\begin{equation}
		P_{\rm E}= -\frac{\Upsilon_2}{{2\Upsilon_1}}-\frac{\sqrt{{\Upsilon_2^2}-4{\Upsilon_1}\Upsilon_3}}{2\Upsilon_1}.
\end{equation}
However, on the other hand, when $\Upsilon_3\le 0$ is satisfied, the function $f(P_{\rm E})$ is always smaller than zero, resulting in a monotonically decreasing secrecy rate.

\section{Proof of Proposition \ref{Le:BR}}\label{App4}
	To prove Proposition \ref{Le:BR}, We only need to find the distance whose right boundary of the beam-depth interval is equal to the Rayleigh distance. As such, we have
		\begin{equation}\label{Eq:huxiang}
	\frac{r_{\rm DT} r_{\rm BD}}{r_{\rm BD}-r_{\rm DT}}=r_{\rm Ray} 	\Longrightarrow 	r_{\rm DT} = \frac{N^2\lambda(1-\theta_{{\rm B},k}^2)}{8\beta_{\varphi}^2+2(1-\theta_{{\rm B},k}^2)}.
	\end{equation}
	Moreover, it can be easily verified that the maximum value is obtained at $\theta_{{\rm B},k}=0$.
Combining the above leads to the desired result.
\vspace{-9pt}
\section{Proof of Proposition \ref{Pro:two_set}}\label{App5}
For ease of exposition, we assume $r_{{\rm B}, i} \le r_{{\rm B}, j}$, and then we prove Proposition \ref{Pro:two_set} from both aspects of beam width and beam depth. First,
	from the beam-width perspective, if user $i$ is located within the beam width of user $j$, user $j$ is definitely located within the beam width of user $i$ as well. This is because the beam width is monotonically increasing with the distance, thus resulting in $\text{BW}_{{\rm \varphi dB},j}\ge\text{BW}_{{\rm \varphi dB},i}$. Next, from the beam-depth perspective, assume users $i$ and $j$ are located at the same angle, i.e., $\theta_{{\rm B}, i} = \theta_{{\rm B}, j}$. 
	In this case, if user $j$ locates within the beam depth of user $i$, i.e., $r_{{\rm L},i}\le r_{{\rm B},j}\le r_{{\rm R},i}$, we have 
	\begin{equation}\label{Eq:huxiang1}
		r_{{\rm B}, j} \le \frac{r_{{\rm B}, i} r_{\rm BD}}{r_{\rm BD}-r_{{\rm B}, i}} 	\Longrightarrow r_{{\rm B}, i} \ge \frac{r_{{\rm B}, j} r_{\rm BD}}{r_{\rm BD}+r_{{\rm B}, j}}.
	\end{equation}
	Notice that the right hand of \eqref{Eq:huxiang1} is exactly equal to $r_{{\rm L},j}$. Therefore, user $i$ is located within the beam depth of user $j$ as well, thus completing the proof.
\end{appendices}
\bibliographystyle{IEEEtran}
\bibliography{Ref_NF_PLS}

\end{document}

%% file: header.tex
\newtheorem{definition}{\emph{\underline{Definition}}}

\newtheorem{lemma}{\emph{\underline{Lemma}}}

\newtheorem{proposition}{\emph{\underline{Proposition}}}

\newtheorem{example}{\bf Example}
\newtheorem{remark}{\bf \emph{\underline{Remark}}}

\def\sinc{\mathrm{sinc}}

\def\({\left(}
\def\){\right)}

\def\b0{{\mathbf{0}}}

\newcommand{\nn}{\nonumber}